\documentclass{article}
\usepackage{hyperref}
\usepackage{float}
\usepackage[all]{xy}
\usepackage[latin1]{inputenc}
\usepackage{epsfig, amssymb, amsmath}
\usepackage{graphics, graphicx}
\usepackage{latexsym,amsfonts,amsmath,amssymb}
\usepackage{url}
\usepackage{relsize}
\usepackage{mathabx}
\usepackage{verbatim}
\usepackage{color} 

\newcommand{\caH}{{\cal H}}
\newcommand{\caI}{{\cal I}}

\newcommand{\caP}{{\cal P}}

\newcommand{\caT}{{\cal T}}
\newcommand{\caX}{{\cal X}}

\newcommand{\pair}[1]{\ensuremath{\langle}{#1}\ensuremath{\rangle}}
\newcommand{\semantics}[1]{{[\![#1]\!]}} 

\newcommand{\bvar}[1]{\mathit{shVar}(#1)}

\newcommand{\wf}{\mathit{wf}}

\newcommand{\chV}[1]{{#1}_\mathit{?}}
\newcommand{\NPASymbols}{\Symbols_\mathit{{SS}_\mathcal{P}}}
\newcommand{\NPAEq}{E_\mathit{{SS}_\mathcal{P}}}
\newcommand{\CstrEq}{E_\mathit{{SS}_\mathcal{P}}}
\newcommand{\SSEq}{E_\mathit{SS}}
\newcommand{\SSSymbols}{\Symbols_\mathit{{SS}}}
\newcommand{\PSymbols}{\Symbols_\caP}
\newcommand{\PEq}{E_{\caP}}
\newcommand{\PAEq}{E_\mathit{PA}}
\newcommand{\PAPEq}{E_\mathit{PA_\caP}}
\newcommand{\SpecPA}{\caP_\mathit{PA}}
\newcommand{\ProcPA}{P_\mathit{PA}}

\newcommand{\PASymbols}{\Symbols_\mathit{{PA}_\caP}}
\newcommand{\BNFSymbols}{\Symbols_\mathit{PA}}
\newcommand{\PAStateSymbols}{\Symbols_\mathit{{PA_\caP+State}}}
\newcommand{\PARls}{R_{\mathit{PA_\caP}}}
\newcommand{\FWRls}{R_\mathit{F_\caP}}
\newcommand{\SpecCstrFW}{\caP_\mathit{CstrSS}}
\newcommand{\StrCstrFW}{P_\mathit{CstrSS}}
\newcommand{\PCstrSymbols}{\Symbols_\mathit{CstrSS_\caP}}
\newcommand{\PCstrFWRls}{R_\mathit{CstrF_\caP}}

\newcommand{\PBWRls}{R_\mathit{B\mathcal{_\caP}}^{-1}}

\newcommand{\PCstrBWRls}{R_\mathit{CstrB_\caP}^{-1}}

\newcommand{\toCstrFW}{\textit{toCstrSS}}
\newcommand{\toCstrFWa}{\textit{toCstrSS*}}

\newcommand{\Pst}{\mathit{Pst}}
\newcommand{\Fst}{\mathit{Fst}}
\newcommand{\nil}{\mathit{nilP}}
\newcommand{\CstrS}{\mathit{CstrS}}

\newcommand{\var}[1]{\mathit{Var}(#1)}
\newcommand{\Variables}{\caX}
\newcommand{\Symbols}{\Sigma}
\newcommand{\TermsOn}[5]{{\caT^{#4}_{#1}(#2)}{}_{#3}^{#5}}
\newcommand{\Terms}{\TermsOn{\Symbols}{\Variables}{}{}{}}

\newcommand{\ETerms}{\TermsOn{\Symbols\!/\!E}{\Variables}{}{}{}}
\newcommand{\TermsS}[1]{\TermsOn{\Symbols}{\Variables}{\sort{#1}}{}{}}

\newcommand{\subterm}[2]{#1|_{#2}}
\newcommand{\sort}[1]{\ensuremath{\mathsf{#1}}}

\newcommand{\cnt}[1]{{\langle #1\rangle}}
\newcommand{\nI}[1]{\ensuremath{#1{\notin}{\cal I}}}
\newcommand{\inI}[1]{\ensuremath{#1{\in}{\cal I}}}

\newcommand{\GTermsOn}[2]{\caT^{#2}_{#1}}
\newcommand{\GTerms}{\GTermsOn{\Symbols}{}}
\newcommand{\GTermsS}[1]{\GTermsOn{\Symbols,\sort{#1}}{}}

\newcommand{\SubstOn}[2]{{\cal S}ubst(#1,#2)}
\newcommand{\Substs}{\SubstOn{\Symbols}{\Variables}{}{}{}}
\newcommand{\idsubst}{\textit{id}}

\newcommand{\composeSubst}{}
\newcommand{\composeRel}{;}
\newcommand{\compose}{\composeSubst}

\newcommand{\congr}[1]{=_{#1}}

\newcommand{\csu}[3]{\textit{CSU}_{#3}({#1})}
\newcommand{\csuV}[3]{\textit{CSU\/}^{#2}_{#3}({#1})}

\newcommand{\funocc}[1]{\mathit{Pos}_{\Symbols}(#1)}

\newcommand{\replace}[3]{#1[#3]_{#2}}
\newcommand{\domain}[1]{\mathit{Dom}(#1)}
\newcommand{\range}[1]{\intrvar{#1}}
\newcommand{\intrvar}[1]{\mathit{Ran}(#1)}

\newcommand{\rewrite}[1]{\rightarrow_{#1}}
\newcommand{\rewrites}[1]{\rightarrow^*_{#1}}

\newcommand{\narrowG}[3]{\mathop{\stackrel{#1}{\leadsto}\!\mbox{}_{#2}^{#3}}}
\newcommand{\narrow}[2]{\mathop{\stackrel{#1}{\leadsto}_{#2}}}

\newcommand{\narrowleft}[2]{\mathop{\stackrel{#1}{\leftsquigarrow}_{#2}}}
\newcommand{\narrowleftTop}[1]{\mathop{\stackrel{#1}\leftsquigarrow}}

\newcommand{\GETerms}{\GTermsOn{\Symbols/E}{}}

\newcommand{\HPsFs}{\caH_\mathit{PS\_FS}}
\newcommand{\HState}{\caH}
\newcommand{\HProc}{\caH_\mathit{LP\_Str}}
\newcommand{\eqname}[1]{\tag{#1}}
\newcommand{\MaxProcId}{\textit{MaxProcId}}
\newcommand{\MaxStrId}{\textit{MaxStrId}}

\newcommand{\IK}{\textit{IK}}

\newcommand{\ProcConf}{\mathit{ProcConf}} 
\newcommand{\LProc}{\mathit{LProc}} 
\newcommand{\Role}{\mathit{Role}} 
\newcommand{\Proc}{\mathit{Proc}} 
 
\newcommand{\Msg}{\mathit{Msg}} 
\newcommand{\Cond}{\mathit{Cond}} 

\usepackage{amsthm}
\newtheorem{theorem}{Theorem}
\newtheorem{lemma}{Lemma}
\newtheorem{definition}{Definition}
\newtheorem{example}{Example}[section]
\newtheorem{remark}{Remark}

\begin{document}







%


\setlength{\pdfpageheight}{\paperheight}
\setlength{\pdfpagewidth}{\paperwidth}


\title{Strand Spaces with Choice\\ via a Process Algebra Semantics 
}

\author{
Fan Yang\\
{\small University of Illinois at Urbana-Champaign, USA}\\
{\small\tt fanyang6@illinois.edu}
\and
Santiago Escobar\\
{\small       Universitat Polit\`ecnica de Val\`encia, Spain}\\
       {\small\tt sescobar@dsic.upv.es}
\and
Catherine Meadows\\
{\small       Naval Research Laboratory, Washington DC, USA}\\
       {\small\tt meadows@itd.nrl.navy.mil}
\and
Jos{\'e} Meseguer\\
{\small       University of Illinois at Urbana-Champaign, USA}\\
       {\small\tt meseguer@illinois.edu}
}       


\date{}
\maketitle

\begin{abstract}
Roles in cryptographic protocols do not always have a linear execution, but may
include choice points causing the protocol to continue along different
paths.  In this paper we address the problem of representing choice in
the strand space model of cryptographic protocols, particularly as it
is used in the Maude-NPA cryptographic protocol analysis tool.

To achieve this goal, we develop and give formal semantics to a
process algebra for cryptographic protocols that supports a rich
taxonomy of choice primitives for composing strand spaces.
In our taxonomy, deterministic and non-deterministic choices are broken down
further.  Non-deter\-ministic choice can be either \emph{explicit}, i.e., one of
two paths is chosen, or \emph{implicit}, i.e., the value of a variable is
chosen non-deterministically.  Likewise, deterministic choice can be
either an \emph{explicit if-then-else} choice, i.e., one path is chosen if a predicate is
satisfied, while the other is chosen if it is not, or 
\emph{implicit deterministic choice}, i.e., execution continues only if a certain pattern is
matched.  We have identified a class of choices which includes finite branching
and some cases of infinite branching, which we address in this paper.

We provide a bisimulation result between
the expected forwards execution semantics of the new process algebra
and
the original symbolic backwards semantics of Maude-NPA
that preserves attack reachability.
%
We have fully integrated  the process algebra syntax
and its transformation into strands in Maude-NPA. 
We illustrate its
expressive power and naturalness with various examples, and
show how it can be effectively used in formal analysis.
This allows users to write protocols from now on 
using the process syntax, which is more convenient for
expressing choice than the strand space syntax, in which choice can only be specified implicitly,
via two or more strands that are identical until the choice point.

 \end{abstract}

 \section{Introduction}

Formal analysis of cryptographic protocols has become 
one of the most successful applications of formal methods to security, with a number of tools available and 
many 
successful applications to the analysis of protocol standards.
In the course of developing these tools it has become clear 
that there are certain universal
features that can best be handled by accounting for them directly in
syntax and semantics of the formal specification language, e.g., unguessable nonces, communication across a network controlled by an attacker, and support for the equational properties of cryptographic primitives. Thus a number of different languages have been developed that include these features.  

At the same time, it is necessary to provide support for more commonly used constructs,
such as
\emph{choice points} that cause the protocol to continue in different
ways, and to do so in such a way that  they are well integrated with  the more specifically cryptographic features of the language.  However, in their original form most of these languages do not support choice, or support it only in a limited way.   

In particular the strand space model \cite{strands}, one of the most popular models designed
for use in cryptographic protocol analysis, does not support choice in its original form; strands describe linear sequences of input
and output messages, without any branching.  One response to dealing with this limitation, and to formalizing
strand spaces in general  has been to embed the strand space model in some other
 formal system that supports choice, e.g., event-based models for concurrency \cite{DBLP:conf/fsttcs/CrazzolaraW02}, Petri nets \cite{DBLP:journals/entcs/Froschle09}, or multi-set
 rewriting \cite{DBLP:conf/csfw/CervesatoDMLS00}. However, we believe that there is an advantage in introducing choice in the strand space
 model itself, while proving soundness and completeness with another formal system in order to validate the augmented model. This allows us to concentrate
 on handling the types of choice that commonly arise in cryptographic protocols.  
A detailed discussion of related work can be found at Section~\ref{sec:related-work}. 
 
\subsection{Contributions} 

This paper is an extended version of the conference paper \cite{YEMMS16}. 
We address the problem of representing choice in the strand space model, particularly as it
is used in the Maude-NPA cryptographic protocol analysis tool.
We have identified a class of choices which includes finite branching
and some cases of infinite branching.
At the theoretical level, we provide a bisimulation result between
the expected forwards execution semantics of the new process algebra
and the original symbolic backwards semantics of Maude-NPA.
This requires extra intermediate forwards and backwards semantics that are included in this paper, together with all the proofs,
but were not included in the conference paper \cite{YEMMS16}.
What these results make possible is a sound and complete symbolic
reachability analysis method for cryptographic protocols with choice
\emph{modulo} equational properties of the cryptographic functions
satisfying the finite variant property (FVP) (see \cite{MNPAmanual3.0} for
a detailed explanation of how FVP theories are supported in Maude-NPA).
At the tool level, we have fully integrated the process algebra syntax,
and its transformation into strands, and have developed new methods
to specify attack states using the process notation  in the recent release of Maude-NPA
3.0 (see Section \ref{Integration}, and \cite{MNPAmanual3.0}).
None of this was available at the time of the conference paper \cite{YEMMS16}.
Furthermore, we illustrate the
expressive power and naturalness of adding choice to strand spaces
with various examples, and
show how it can be effectively used in formal analysis.

\subsection{Choice in Maude-NPA}
 
Previous to this work, Maude-NPA offered some ways of handling
choice, but its scope was limited, and a uniform semantics of choice
was lacking.  Several  kinds of branching could be handled by a protocol
composition method  
in which a single parent strand is composed with one or more
child strands.  Although protocol composition is intended for
modular construction
of protocols, with suitable restrictions it can also be
used to express both non-deterministic branching and deterministic
branching predicates on pattern matching of output parameters of the
parent with the input parameters of the child.  However, repurposing
composition to branching has its limitations.  First of all, it is
possible to inadvertently introduce non-deterministic choice into what
was intended to be deterministic choice by unwise choice of input and
output parameters.  Secondly, the limitation to pattern matching rules
out certain types of deterministic choice conditioned on predicates
that cannot be expressed this way, e.g., disequality predicates.
Finally, implementation of choice via composition can be
inefficient, since Maude-NPA must evaluate all possible child strands
that match a parent strand.

Maude-NPA, in common with many other cryptographic protocol analysis tools,  also offers a type of implicit choice that does not involve branching: non-deterministic choice of the values
of certain variables.  For example,   a strand that describes an initiator communicating with a responder generally uses variables for both the initiator and responder names; this represents a non-deterministic choice of initiator responder identities.  However, the semantic implications of this kind of choice were not that well understood, which made it difficult to determine where it could safely be used.  Clearly, a more unified treatment of choice was necessary, together with a formal semantics of choice.

 In support of this work we have developed a taxonomy of choice
 in which the categories of deterministic and non-deterministic choice are further subdivided.  First of all, \allowbreak non-deterministic choice
 is subdivided into \emph{explicit} and \emph{implicit}
 non-deterministic choice.  In explicit non-deterministic choice a
 role\footnote{As further explained later, the behaviors
   of protocol participants, e.g., sender, receiver, server,
   etc., are described by their respective \emph{roles}.  Since a
   protocol may have multiple sessions, various participants, called
\emph{principals}, may play
   the same role in different sessions.}
 chooses either one
 branch or another at a choice point \ non- deterministically.  In implicit non-deterministic choice a logical \emph{choice variable} is introduced which may be non-deterministically instantiated
 by the role.   Deterministic choice is subdivided into (explicit) \emph{if-then-else} choice and \emph{implicit deterministic choice}.  In if-then-else
 choice a predicate is evaluated.    If the predicate evaluates to true one branch is chosen, and if it evaluates to false another branch is chosen. 
 Deterministic choice with more than two choices can be modeled by 
nesting of if-then-else choices. In implicit deterministic choice, a term pattern is used as an implicit guard, so that only messages matching such pattern can be chosen i.e., accepted, by the role. 
 Although implicit deterministic choice can be considered a special case of if-then-else choice in which the
 second branch is empty, it is often simpler to treat it separately.   Classifying choice in this way allows us to represent all possible behaviors of a protocol by a finite number
 of strands modeling possible executions, while still allowing the variables
 used in implicit non-deterministic and deterministic choice to be instantiated in an infinite number of ways.

 \subsection{A Motivating Example}\label{sec:ex}

In this section we introduce a protocol that we will use as a running example.
It is a simplified version of the handshake protocol in TLS 1.3
\cite{tls1.3.12},
 a proposed update to the TLS standard for client-server authentication.   This protocol, like most other protocol standards,
offers a number of different choices that 
are applied in different situations.  In order to make the presentation and discussion manageable, 
we present only a subset here:  the client chooses a  Diffie-Hellman group, and proposes it to the server.  The server can either accept it or request that the client
proposes a different group.
In addition, the server has the option of requesting that the client authenticates itself.  We present the protocol at a high level similar to the style used in \cite{tls1.3.12}.

 \begin{example}\label{ex:tlsinf} 
 We let a dashed arrow $\dasharrow$ denote an optional message, and an asterisk  * denote an optional field.
 	\begin{enumerate}
		\item $C \rightarrow S:$  ClientHello, Key\_Share \\
		The client sends a Hello message containing a nonce and the Diffie-Hellman group it wants to use.  It also sends a Diffie-Hellman key share.  
		\begin{itemize}
			\item{1.1} $S \dashrightarrow C:$  HelloRetryRequest \\
		The server may optionally reject the Diffie-Hellman group proposed by the client and request a new one.	
			\item{1.2} $C \dashrightarrow S:$  DHGroup, Key\_Share \\
		The client proposes a new group and sends a new key share.  
		\end{itemize} 
		\item $S \rightarrow C:$ ServerHello, Key\_Share, \\\mbox{}\hspace{8ex} \{AuthReq*\},\{CertificateVerify\}, \{Finished\}\\
		The server sends its own Hello message and a Diffie-Hellman key share.  It may optionally send an AuthReq to the client to authenticate itself
		with a public key signature from its public key certificate.  It then signs the entire handshake using its own public key in the CertificateVerify field.
		Finally, in the Finished field it computes a  MAC over the entire handshake using the shared Diffie-Hellman key.  
		The $\{\}$ notation denotes a field encrypted using the shared Diffie-Hellman key.
		\item $C \rightarrow S:$  $\{$CertificateVerify*$\}$, $\{$Finished$\}$ \\
		If the client received an AuthReq from the server it returns its own CertificateVerify and Finished fields.
	\end{enumerate}
 \end{example}

\subsection{Plan of the paper}  

The rest of the paper is organized as follows.
After some preliminaries in Section \ref{sec:prelim} and a high level introduction of the Maude-NPA tool in Section \ref{sec:mnpa-syntax}, we first define the process
algebra syntax and operational semantics in Section \ref{sec:PA}. 
In Section \ref{sec:Cstr} we extend Maude-NPA's strand space syntax to include choice operators.
The main bisimulation results between the expected forwards semantics of the process algebra in Section \ref{sec:PA}
and the original symbolic backwards strand
 semantics of Maude-NPA of Section \ref{sec:mnpa-syntax} 
are Theorems~\ref{thm: Soundness-PA}~and~\ref{thm: Completeness-PA}.
They are proved 
by introducing an intermediate semantics, a forward strand space semantics originally introduced in \cite{Escobar2014FW}. 
First, in Section~\ref{sec:Cstr} we extend the strand space model with constraints, since strands are the basis
of both the forwards semantics and the backwards semantics of Maude-NPA.
In Section~\ref{sec:CstrFW} we 
augment
the forwards strand space semantics of \cite{Escobar2014FW} 
with choice operators and operational semantic rules to produce a \emph{constrained forwards semantics}.
In Section~\ref{sec:PA-CstrFW-bisim} we prove bisimilarity of
the process algebra semantics of Section \ref{sec:PA}
and the constrained forwards semantics of Section \ref{sec:CstrFW}.
In~\cite{Escobar2014FW} 
the forwards strand space semantics
was proved sound and complete w.r.t. the original symbolic backwards semantics of Maude-NPA
and, therefore, such proofs had to be extended to handling constraints.
In Section~\ref{sec:CstrBW} we 
augment
the original symbolic backwards semantics of Maude-NPA
with choice operators and operational semantic rules to produce a
\emph{constrained backwards semantics}.
In Section~\ref{sec:proof-CstrFW-CstrBW}
we then prove  
that the constrained backwards semantics is sound and complete
with respect to the constrained forwards semantics.
By combining the bisimulation
between the process algebra and the constrained forwards semantics on the one hand
(Theorem~\ref{thm:Bisimulation})
and the bisimulation 
between the constrained forwards semantics and the constrained backwards semantics on the other hand 
(Theorems~\ref{thm:soundness} and~\ref{thm:completeness})
we obtain the main bisimulation results
(Theorems~\ref{thm: Soundness-PA} and~\ref{thm: Completeness-PA}).
Finally, in Section \ref{sec:Exp} we describe 
how the process algebra has been fully integrated into Maude-NPA
and show  some 
experiments we have run using Maude-NPA on various protocols
exhibiting both deterministic and non-deterministic choice. 
In Section \ref{sec:related-work}
we  discuss related and future work, in particular the potential of using the process algebra syntax as a specification language.
Finally,
we conclude in Section~\ref{sec:Conclusion}. 
\section{Preliminaries}\label{sec:prelim}

We follow the classical notation and terminology 
for term rewriting 
and 
for rewriting logic and order-sorted notions, see \cite{Meseguer92}.
%
We assume an order-sorted signature ${\bf \Sigma} = (\sort{S}, \leq, \Sigma)$
with poset of sorts $(\sort{S}, \leq)$.  
We also assume an $\sort{S}$-sorted family
$\Variables=\{\Variables_\sort{s}\}_{\sort{s} \in \sort{S}}$
of disjoint variable sets with each $\Variables_\sort{s}$
countably infinite.
$\TermsS{\sort{s}}$
is the set of terms of sort \sort{s},
and
$\GTermsS{\sort{s}}$ is the set of ground terms of sort \sort{s}.
We write
$\Terms$ and $\GTerms$ for the corresponding order-sorted term algebras.
For a term $t$, $\var{t}$ denotes the set of variables in $t$.


A \textit{substitution} $\sigma\in\Substs$ is a sorted mapping from a finite
subset of $\Variables$ to $\Terms$.
Substitutions are written as
$\sigma=\{X_1 \mapsto t_1,\ldots,X_n  \mapsto t_n\}$ where
the domain of $\sigma$ is
$\domain{\sigma}=\{X_1,\ldots,X_n\}$
and
the set
of variables introduced by terms $t_1,\ldots,\allowbreak t_n$ is written $\range{\sigma}$.
The identity
substitution is denoted $\idsubst$.  Substitutions are homomorphically extended
to $\Terms$.
The application of a substitution $\sigma$ to a term $t$ is
denoted by $t\sigma$. 
For simplicity, we assume that every substitution is idempotent,
i.e., $\sigma$ satisfies $\domain{\sigma}\cap\range{\sigma}=\emptyset$.
This ensures $t\sigma=(t\sigma)\sigma$.
The restriction of $\sigma$ to a set of variables
$V$ is $\subterm{\sigma}{V}$. 
Composition of two substitutions $\sigma$ and $\sigma'$ is denoted by $\sigma\compose\sigma'$.
 A substitution $\sigma$ is a ground substitution if $\range{\sigma} =\emptyset$.

A \textit{$\Symbols$-equation} is an unoriented pair $t = t'$, where
$t,t' \in \TermsS{s}$ for some sort $\sort{s}\in\sort{S}$.  Given
$\Symbols$ and a set $E$ of $\Symbols$-equations, 
order-sorted
equational logic induces a congruence relation $\congr{E}$ on terms
$t,t' \in \Terms$. 
The $E$-equivalence class of a term $t$ is denoted by $[t]_E$
and 
$\ETerms$ and $\GETerms$ denote the corresponding order-sorted term algebras
modulo $E$.
Throughout this paper we
assume that $\GTermsS{s}\neq\emptyset$ for every sort \sort{s},
because this affords a simpler deduction system.
An \emph{equational theory} $(\Symbols,E)$ is a
pair with $\Symbols$ an order-sorted signature and $E$ a
set of $\Symbols$-equations.
%
%
The \textit{$E$-subsumption} preorder $\sqsupseteq_{E}$ (or just $\sqsupseteq$ if $E$
is understood) holds between
$t,t' \in \Terms$, denoted $t \sqsupseteq_{E} t'$ (meaning that $t$ is
\emph{more general} than $t'$ modulo $E$), if there
is a substitution $\sigma$ such that
$t\sigma \congr{E} t'$; such a substitution $\sigma$ is said to be
an \textit{$E$-match} from $t'$ to $t$.

An \textit{$E$-unifier} for a $\Symbols$-equation $t = t'$ is a
substitution $\sigma$ such that $t\sigma \congr{E} t'\sigma$.  For
$\var{t}\cup\var{t'} \subseteq W$, a set of substitutions $\csuV{t =
  t'}{W}{E}$ is said to be a \textit{complete} set of unifiers
for the equality $t = t'$ modulo $E$ away from $W$ iff:
(i) each $\sigma \in
\csuV{t = t'}{W}{E}$ is an $E$-unifier of $t = t'$;
(ii) for
any $E$-unifier $\rho$ of $t = t'$ there is a $\sigma \in
\csuV{t=t'}{W}{E}$ such that $\subterm{\sigma}{W} \sqsupseteq_{E} \subterm{\rho}{W}$; 
(iii) for all
$\sigma \in \csuV{t=t'}{W}{E}$, $\domain{\sigma} \subseteq
(\var{t}\cup\var{t'})$ and $\range{\sigma} \cap W = \emptyset$.
If the set of variables $W$ is irrelevant or is understood from the context,
we write $\csu{t = t'}{W}{E}$ instead of $\csuV{t = t'}{W}{E}$.
An
$E$-unification algorithm is \textit{complete} if for any equation $t
= t'$ it generates a complete set of $E$-unifiers.  
A unification algorithm is said to be
\textit{finitary} and complete if it always terminates after
generating a finite and complete set of solutions.


A \textit{rewrite rule} is an oriented pair $l \to r$,
where\footnote{We do not impose the requirement
$\var{r} \subseteq \var{l}$, since extra variables
(e.g., choice variables)  may be introduced in the righthand side of a rule.
Rewriting with extra variables in righthand sides is handled by
allowing the matching substitution to instantiate these extra
variables in any possible way. Although 
this may produce an infinite number of one-step concrete rewrites from a term
due to the infinite number of possible instantiations, the 
symbolic, narrowing-based analysis 
used by Maude-NPA and explained below can
cover all those infinite possibilities in a finitary way.
}
$l \not\in \Variables$ 
and $l,r \in \TermsS{s}$ for some
sort $\sort{s}\in\sort{S}$.
An \textit{(unconditional) order-sorted
rewrite theory} is a triple $(\Symbols,E,R)$ with $\Symbols$ an
order-sorted signature, $E$ a set of $\Symbols$-equations, and $R$ a set
of rewrite rules.

The rewriting relation on
$\Terms$, written $t \rewrite{R} t'$  or
$t \rewrite{p,R} t'$
holds between $t$ and $t'$ iff
there exist $p \in \funocc{t}$,
$l \to r\in R$ and a substitution $\sigma$, such that
$\subterm{t}{p} = l\sigma$,
and $t' = \replace{t}{p}{r\sigma}$.
The subterm $t|_{p}$ is called a \emph{redex}.
The relation $\rewrite{R/E}$
on $\Terms$ is
${\congr{E} \composeRel\rewrite{R}\composeRel\congr{E}}$,
i.e.,
$t \rewrite{R/E} t'$
iff there exists $u,u'$ s.t.
${t \congr{E} u \rewrite{R} u' \congr{E} t'}$.
Note that
$\rewrite{R/E}$ on $\Terms$
induces a relation
$\rewrite{R/E}$ on the free $(\Sigma,E)$-algebra ${\TermsOn{\Symbols\!/\!{E}}{\Variables}{}{}{}}$
by
${[t]_{E} \rewrite{R/E} [t']_{E}}$ iff $t \rewrite{R/E} t'$.
The transitive (resp. transitive and reflexive) closure of $\rewrite{R/E}$ is denoted
$\rewrite{R/E}^+$ (resp. $\rewrites{R/E}$).


The
$\rewrite{R/E}$ relation can be difficult to compute.  However, under the appropriate
conditions it is equivalent to the $R{,}E$ relation 
in which it is enough to compute the relationship on any representatives of two
$E$-equivalence classes.
A relation $\rewrite{R,E}$ on $\Terms$ is defined as:
  $t \rewrite{p,R,E} t'$ (or just $t \rewrite{R,E} t'$)
  iff there exist $p \in \funocc{t}$,
  a rule $l \to r$ in $R$,
  and a substitution $\sigma$ such that
  $\subterm{t}{p} \congr{E} l\sigma$
  and $t' = \replace{t}{p}{r\sigma}$.

 Let $t$ be a term and $W$ be a set of variables such that
  $\var{t} \subseteq W$,
  the $R,E$-\textit{narrowing} relation on $\Terms$ is defined as 
  $t \leadsto_{p,\sigma,R,E} t'$ 
  ($\leadsto_{\sigma,R,E}$ if $p$ is understood, 
  $\leadsto_{\sigma}$ if $R,E$ are also understood,
  and
   $\leadsto$ if $\sigma$ is also understood) 
  if there is a non-variable position $p \in \funocc{t}$, 
  a rule $l \to r \in R$ properly renamed 
  s.t. $(\var{l}\cup\var{r})\cap W=\emptyset$, 
  and 
  a unifier
   $\sigma \in \csuV{\subterm{t}{p}=l}{W'}{E}$ 
  for $W' = W\cup\var{l}$,
  such that
  $t'=(\replace{t}{p}{r})\sigma$.
For convenience, 
in each narrowing step $t \leadsto_{\sigma} t'$
we only specify the part of $\sigma$ that binds variables of $t$.
The transitive (resp. transitive and reflexive) closure of $\leadsto$ is denoted by
$\leadsto^+$ (resp. $\leadsto^*$).
We may write $t \narrowG{}{\sigma}{k} t'$ 
if there are
$u_{1},\ldots,u_{k-1}$ and substitutions $\rho_{1},\ldots,\rho_{k}$
such that $t \narrow{}{\rho_{1}} \allowbreak u_{1} \allowbreak  \cdots \allowbreak u_{k-1}
\narrow{}{\rho_{k}} t'$, $k \geq 0$, and $\sigma =
\rho_{1}\cdots\rho_{k}$.  

Maude-NPA uses \emph{backwards narrowing} (i.e., uses protocol  rules
$l \rightarrow r$ ``in reverse'' as rules $r \rightarrow l$)
\emph{modulo} the algebraic properties of cryptographic functions
as a \emph{sound and complete} reachability analysis method.
Section \ref{sec:proof-CstrFW-CstrBW} gives a detailed
proof of the soundness and completeness of this
method for strands with choice.



\section{Overview of Maude-NPA}\label{sec:mnpa-syntax}
 
 Here we give a high-level summary of Maude-NPA.
For further details please see
\cite{MNPAmanual3.0}.
 
 
 %

 
 Given a protocol $\mathcal{P}$,
we define  its specification in the strand space model 
as a rewrite theory of the form $(\NPASymbols,\NPAEq, \PBWRls)$,
where 
(i)
the signature $\NPASymbols$ is split into predefined symbols $\SSSymbols$ for strand syntax
and user-definable symbols $\PSymbols$ based on a parametric sort
\sort{Msg} of messages,
(ii)
the algebraic properties $\NPAEq$ are also split
into the algebraic properties of the strand notation $\SSEq$
and
the user-definable algebraic properties $\PEq$ for the cryptographic functions,
and 
(iii) the transition rules $\PBWRls$ are defined on states, i.e.,
terms of a predefined sort \sort{State}.  They are \emph{reversed}
for backwards execution.


In Maude-NPA 
states are modeled as
elements of an initial algebra $\mathcal{T}_{\NPASymbols/\NPAEq}$,
i.e.,
an $\NPAEq$-equivalence class $[t]_{\NPAEq}\in
\caT_{\NPASymbols/ \NPAEq}$ with $t$ a ground
$\NPASymbols$-term.
%
%
A state 
has the form
$\{ S_1 \,\&\, \cdots \,\&\, S_n \,\&\, \{\textit{IK}\}\}$
where $\&$ is an associative-commutative union 
operator with identity symbol $\emptyset$. 
Each element in
the set is either a \emph{strand} 
$S_i$ or the \emph{intruder knowledge} $\{\textit{IK}\}$ at that state.

The \emph{intruder knowledge} $\{\textit{IK}\}$
belongs to the state and
is represented as a set of facts 
using the comma as an associative-commutative union operator 
with identity element $empty$. 
There are two kinds
of intruder facts: \emph{positive} knowledge facts (the intruder knows
message $m$,
i.e., $\inI{m}$), and \emph{negative} knowledge facts (the intruder \emph{does not
yet know} $m$ but \emph{will know it in a future state}, i.e., $\nI{m}$),
where $m$ is a message expression.

A \emph{strand}
\cite{strands} specifies the sequence of messages sent and received
by a principal executing a given role in the protocol and is represented as a sequence
of messages 
$
[\textit{msg}_1^\pm, \textit{msg}_2^\pm, \textit{msg}_3^\pm,
 \ldots,\allowbreak \textit{msg}_{k-1}^\pm, \textit{msg}_k^\pm]
 $ 
 with 
$\textit{msg}_{i}^\pm$ 
either
$\textit{msg}_i^-$ (also written $-\textit{msg}_i$) representing an input
message, 
or
$\textit{msg}_i^+$ (also written $+\textit{msg}_i$) representing an output message.
Note that 
each $\textit{msg}_{i}$ is a term 
of a special sort \textsf{Msg}.

Strands are used to represent both the actions of honest principals (with a strand 
specified for each protocol role) 
and the actions of an intruder 
(with a strand 
for each action an intruder is able to perform on messages).
In Maude-NPA strands evolve over time;
the symbol $|$ is used to divide past and future.
That is, given a strand 
$[\ \textit{msg}_1^{\pm},\ \ldots,\ \textit{msg}_i^{\pm}\ |\allowbreak\ \textit{msg}_{i+1}^{\pm},\ \ldots,\ \textit{msg}_k^{\pm}\ ]$,
 messages $\textit{msg}_1^\pm,\linebreak[2] \ldots, \textit{msg}_{i}^\pm$ are
the \emph{past messages}, and messages $\textit{msg}_{i+1}^\pm, \ldots,\textit{msg}_k^\pm$
are the \emph{future messages} ($\textit{msg}_{i+1}^\pm$ is the immediate future
message).
A strand 
$[\textit{msg}_1^\pm, \linebreak[2]\ldots, \linebreak[2] \textit{msg}_k^\pm]$ 
is shorthand for 
$[nil ~|~ \textit{msg}_1^\pm,\linebreak[2] \ldots,\linebreak[2] \textit{msg}_k^\pm , nil ]$.
An \emph{initial state} is a state where the bar is at the beginning for all strands in the state,
and the intruder knowledge has no fact of the form $\inI{\textit{m}}$. 
A \emph{final state} is a state where the bar is at the end for all strands in the state
and there is no intruder fact of the form $\nI{\textit{m}}$.


Since the number of 
states $T_{\NPASymbols/ \NPAEq}$
is in general infinite, rather than exploring concrete protocol states
$[t]_{\NPAEq}\in T_{\NPASymbols/\NPAEq}$  Maude-NPA 
explores
\emph{symbolic strand state patterns}  
 $[t(x_{1},\ldots,x_{n})]_{\NPAEq} \in
T_{\NPASymbols/ \NPAEq}(\Variables)$ on the free
$(\NPASymbols, \allowbreak \NPAEq)$-algebra over a set of
variables $\Variables$. In this way, a state pattern $[t(x_{1},\ldots,x_{n})]_{\NPAEq}$
represents not a single concrete state but a possibly infinite set of
such states, namely all the \emph{instances} of the pattern
$[t(x_{1},\ldots,x_{n})]_{\NPAEq}$ where the variables $x_{1},\ldots,x_{n}$
have been instantiated by concrete ground terms. 

The semantics of Maude-NPA is expressed in terms of the following 
\emph{forward rewrite rules} that describe how a protocol moves from one state to another 
via the intruder's interaction with it.   

\begin{small}	
		\begin{align}
		&  \{
		SS\ 
		\&\ [L ~|~ M^-, L']\ 
		\&\ \{\inI{M},IK\}\}		
		\to 
		 \{
		\textit{SS}\ 
		\&\ [L, M^- ~|~ L']\ 
		\&\ \{\inI{M},IK\} 
		\} \eqname{-}
		\label{eq:negative-1}
		\\[1ex]
		&   \{  
		SS\ 
		\&\ [L ~|~ M^+, L']\ 
		\&\ \{IK\}
		\} 
		\to 
		 \{  
		SS \ 
		\&\ [L, M^+ ~|~ L']\
		\&\ \{IK\}
		\}   \eqname{+}
		\label{eq:positiveNoLearn-2}
		\end{align}

		\begin{align}
		&   \{  
		SS\ 
		\&\ [L ~|~ M^+, L']\ 
		\&\ \{\nI{M},IK\}
		\} 
		\to 
		  \{  
		SS \ 
		\&\ [L, M^+ ~|~ L']\
		\&\ \{\inI{M}, IK\}
		\}   \eqname{++}
		\label{eq:positiveLearn-4}\\[1ex]
& {	\begin{array}[t]{@{}l@{}}
	\forall \ [ l_1, u^+, l_2] \in \caP: 
	  \{SS \, \& \, [\, l_1 |  \,  u^+, l_2\,] \,\&\, \{\nI{u},IK\}\}  
	\to   \{SS \,  \& \,  \{\inI{u},IK\}\}
	\end{array} 
	\label{eq:newstrand-positive}%
}\eqname{\&}
	\end{align}%
\end{small}

\noindent
where   $L$ and $L'$ are variables denoting a list of strand messages,
$\textit{IK}$ is a variable for a set of intruder facts 
(\inI{m} or \nI{m}),
 $\textit{SS}$ is a variable denoting a set of strands,
 and $l_1$, $l_2$ denote a list of strand messages.
The set $\PBWRls$
of \emph{backwards}  state transition rules is
 defined by reversing the direction of the above set of rules
 $\{ 
    \eqref{eq:negative-1},
    \eqref{eq:positiveNoLearn-2},
    \eqref{eq:positiveLearn-4}\} 
    \cup 
    \eqref{eq:newstrand-positive} $.
In the backwards executions of $\eqref{eq:newstrand-positive},\eqref{eq:positiveLearn-4}$,
$\nI{u}$ marks when the intruder learnt $u$.

One uses Maude-NPA to find an attack by specifying an insecure state pattern called an \emph{attack pattern}.  Maude-NPA attempts to find a path from an initial state to the attack pattern via backwards narrowing (narrowing using the rewrite rules with the orientation reversed). 
That is, a 
narrowing sequence from an initial state to 
an attack state is 
searched \emph{in reverse} as  a \emph{backwards path} from the attack state to the initial state. 
Maude-NPA attempts to find paths until it can no longer form any backwards narrowing steps, at which point it terminates.
If  at that point it has not found an initial state, the attack
pattern is shown to be \emph{unreachable} modulo $\NPAEq$.
(Section \ref{sec:proof-CstrFW-CstrBW} gives a detailed
proof of the soundness and completeness of this symbolic
method for the Maude-NPA extension supporting strands with choice).
 Note that Maude-NPA places
\emph{no bound on the number of sessions}, so reachability is undecidable in general.  
Note also that Maude-NPA does not perform any data abstraction
such as a bounded number of nonces. 
However, the tool makes use of a number of sound and complete state space reduction techniques that help to identify unreachable and redundant states, 
and thus make termination more likely.

\section{A Process Algebra for  Protocols with Choice}\label{sec:PA}

In this section we define a process algebra that extends the strand
space model to naturally specify
protocols exhibiting choice points.
Throughout the paper we refer to this process algebra as the
\emph{protocol process algebra}.

The rest of this section is organized as follows.
First, in Section~\ref{sec:syntaxPA} we define the syntax of the protocol  process algebra
and state the requirements that a \emph{well-formed process} must satisfy. 
Then in Section~\ref{sec:PASpecification}, we explain how \emph{protocol specifications} can be defined in this process algebra.
In Section~\ref{sec:semanticsPA} we then define the 
\emph{operational semantics} of the protocol  process algebra.
Note that the operational semantics of Maude-NPA given in Section~\ref{sec:mnpa-syntax} 
corresponds to a symbolic backwards semantics
while in Section~\ref{sec:semanticsPA} we give a rewriting-based
forwards semantics for process algebra.  
Sections \ref{sec:PA-CstrFW-bisim} and \ref{sec:proof-CstrFW-CstrBW} will
relate these two semantics using \emph{bisimulations}.

\subsection{Syntax of the Protocol Process Algebra }\label{sec:syntaxPA}

In the \emph{protocol process algebra} the behaviors of both honest principals
and the intruder are represented  by \emph{labeled processes}. 
Therefore, a protocol is specified as a set of labeled processes.
Each process performs a sequence of actions, namely, sending or receiving
a message, and may perform deterministic or non-deterministic choices. 
The protocol process algebra's syntax is 
parameterized\footnote{More precisely, as explained in
Section \ref{sec:PASpecification},
they are parameterized
by a user-definable equational theory
$(\Sigma_{\mathcal{P}},E_{\mathcal{P}})$
having a sort \sort{Msg} of messages.}
 by a sort \sort{Msg} of messages and has the following syntax:

\begin{align}
 \ProcConf~ &::= \LProc ~|~ \ProcConf ~\&~ \ProcConf ~|~ \emptyset \notag\\[-.8ex]
 \LProc ~&::= (\Role, I, J)~ \Proc \notag\\[-.8ex]
 \Proc~ &::=  \nil ~ |~ +\Msg~ | ~ -\Msg ~|~ \Proc \cdot \Proc  ~|  \notag \\[-.8ex]
              &~~~~~~~ \Proc~?~\Proc ~ | ~\textit{if} ~\Cond~ \textit{then} ~\Proc~ \textit{else} ~\Proc    \notag\\[-.8ex]
 \Cond~   &::=  \Msg~ \neq~ \Msg ~| ~\Msg =~ \Msg   \notag 
\end{align} 

\begin{itemize}
	
\item $\ProcConf$ stands for a \emph{process configuration}, that is, a set of labeled processes. The symbol \& 
is used to denote set union for sets of labeled processes. 

\item $\LProc$ stands for a \emph{labeled process}, that is, a process $\Proc$ with a label $(\Role,I,J)$.
$\Role$ 
refers to the role of the process in the protocol (e.g., initiator or responder). $I$ is a 
natural number denoting the identity of the process, which distinguishes different instances(sessions) of a process specification.
$J$ indicates that the action at stage $J$ of the process specification 
will be the next one to be executed, that is, the first $J-1$ actions of the process for role $\Role$ have  already been executed. 
Note that we omit $I$ and $J$ in the protocol specification when both $I$ and $J$ are $0$. 

\item $\Proc$ defines the actions that can be executed within a process. ${+\Msg}$, and ${-\Msg}$ respectively denote
 sending out or receiving a message $\Msg$.  We assume a single channel, through which all messages are sent or received by the intruder.  
``$\Proc~\cdot~\Proc$" denotes \emph{sequential composition} of processes,
 where 
symbol \verb!_._! is associative and has the empty process $\nil$ 
as identity.
``$\Proc~?~\Proc$" denotes an explicit \emph{nondeterministic choice}, whereas 
``$\textit{if} ~\Cond~ \textit{then} \allowbreak ~\Proc~ \textit{else}
~\Proc$" denotes an explicit \emph{deterministic choice}, whose
continuation depends on the 
satisfaction of the constraint $\Cond$. 

%

\item $\Cond$ denotes a constraint that will be evaluated in explicit deterministic choices. In this work we only consider
 constraints that are either equalities ($=$) or disequalities ($\neq$)  between message expressions. 

\end{itemize}

\vspace{-1ex}
Let 
$PS,~QS$, and $RS$ be process configurations,
and
$P,~Q$, and $R$ be protocol processes. 
 Our protocol syntax satisfies the following \emph{structural axioms}:

\begin{small}
\begin{minipage}{0.65\linewidth}
\begin{gather}
PS  \, \& \, QS = QS \, \& \, PS  \\
(PS \, \& \, QS) \, \& \, RS = PS \, \& \, ( QS  \, \& \,  RS)\\
(P \, \cdot \, Q) \cdot \, R  =  P \, \cdot \, (Q \, \cdot \, R)
\end{gather}
\end{minipage}
\begin{minipage}{0.30\linewidth}
\begin{gather}
 PS  \&  \emptyset = PS  \\
P \, \cdot \, \nil = P \\
\nil \, \cdot \, P = P 
\end{gather}
\end{minipage}
\end{small}

The specification of the processes defining a protocol's behavior may contain 
some variables denoting information that 
the principal executing the process does not yet know, or that will be different in different executions. 
  In all protocol specifications we assume
three disjoint kinds of variables:
    \begin{itemize}
        	\item \textbf{\emph{fresh variables}}: these are not really variables in the standard sense,
        	 but \emph{names} for \emph{constant values} in a data type \sort{V_{fresh}}
        	 of unguessable values such as nonces. 
        	 A \emph{fresh variable} $r$ is always associated with a role $ro\in Role$ in the protocol.
        	 For each protocol session $i$, we associate to $r$ a unique name $r.ro.i$ for a constant in the data type \sort{V_{fresh}}.
        	 What is assumed is that if $r.ro.i \neq r'.ro'.j$ (including the case $r.ro.i \neq r.ro.j$), 
        	 the values interpreting $r.ro.i$ and $r'.ro'.j$ in \sort{V_{fresh}} are both \emph{different} and \emph{unguessable}.
        	 In particular, for role $ro\in Role$, the interpretation mapping $I: \{r.ro.i \mid i\in \mathbb{N}\} \rightarrow \sort{V_{fresh}}$ 
        	 is \emph{injective} and \emph{random}.
        	 In our semantics, a constant $r.ro.i$ denotes its (unguessable) interpretation $I(r.ro.i) \in \sort{V_{fresh}}$.
%
%
        	Throughout this paper we will denote this kind of variables
        	as $r,r_1,r_2,\ldots$.

        	\item \textbf{\emph{choice variables}}: variables  first
        	appearing in a \emph{sent message} $\mathit{+M}$, which can be substituted 
        	by any value arbitrarily chosen from a possibly infinite domain.
        A choice variable indicates an  \emph{implicit non-deterministic choice}.
        Given a protocol with choice variables, each possible substitution of 
        these variables denotes a possible continuation of the protocol.
        	We always denote choice variables by uppercase
        	letters postfixed with the symbol ``?'' as a subscript, e.g., $\chV{A},\chV{B},\ldots$.

        	\item \textbf{\emph{pattern variables}}:  variables first appearing
        	in a \emph{received message} $\mathit{-M}$. These variables will be instantiated 
        	when matching sent and received messages. 
        	\emph{Implicit deterministic choices} are indicated by
                terms containing pattern variables, 
        	since failing to match a pattern term may lead to the rejection of a message.
        	A pattern term  plays the implicit role of a guard, 
        	so that, depending on the different ways of matching it, the protocol can have different continuations.
		    This kind of variables will be written with uppercase letters, e.g.,
		    $A,B,N_A,\ldots$.
        \end{itemize} 
        
Note that fresh variables are distinguished from other variables by having a specific sort \sort{Fresh}.
Choice variables or pattern variables can never have sort \sort{Fresh}.

To guarantee the requirements on different kinds of variables that
can appear in a given process, we consider
only  \emph{well-formed} processes.
We make this notion precise
by defining  a function
$\mathit{wf}: \mathit{\Proc} \rightarrow \mathit{Bool}$
checking whether a given process is well-formed. 
A labeled process is \emph{well-formed} if the process it labels is well-formed.
A process configuration is  \emph{well-formed} if all the labeled process 
in it are well-formed.
The definition of $\mathit{wf}$ uses an auxiliary function 
$\mathit{shVar} : \mathit{\Proc}  \rightarrow \mathit{VarSet}$,
retrieving the ``shared variables'' of a process, 
i.e., the set of variables that show up in all branches. 
Below we define both functions,  where 
$P,~Q,$ and $R$ are processes, $M$ is a message, 
and $T$ is a constraint.

\vspace{-1.5ex}
\begin{small}
\begin{align*}
& \bvar{+M ~\cdot P} =  \var{M} \cup \bvar{P} \\[-.5ex]
& \bvar{-M ~\cdot P} =   \var{M} \cup \bvar{P}  \\[-.5ex]
& \bvar{ (\textit{if} ~T~ \textit{then} ~P~ \textit{else} ~Q) ~\cdot R} 
= 
    \var{T} \cup (\bvar{P} \cap \bvar{Q}) \cup \bvar{R}  \\[-.5ex]
& \bvar{ (P~?~Q)~\cdot R } =  ( \bvar{P} \cap \bvar{Q}) \cup \bvar{R} \\[-.5ex]
& \bvar{\nil} = \emptyset
\end{align*}

\vspace{-4ex}

\begin{align*}
& \wf(P \cdot +M) = \wf(P)  
  ~~~~~~~~ \textit{if} ~~ (\var{M} \cap   \var{P}) \subseteq \bvar{P} \\[-.5ex]
& \wf( P \cdot -M) = \wf(P)    
~~~~~~~~ \textit{if} ~~ (\var{M} \cap   \var{P}) \subseteq \bvar{P} \\[-.5ex]
& \wf( P \cdot 
 (\textit{if} ~~T~ \textit{then} ~Q~ \textit{else} ~R))
  = \wf(P \cdot Q) \wedge \wf(P \cdot R) 
~~~~  \textit{if} \  P \neq \nil \ \textit{and } \ Q \neq \nil \ \textit{and }  \\[-.5ex]
& \hspace*{80mm}      ~\var{T} \subseteq \bvar{P}\\[-.5ex]
&\wf(P \cdot (Q~?~R) )= \wf(P \cdot Q)  \wedge \wf(P \cdot R) 
~~~~ ~~~~ \textit{ if } Q \neq \nil \ \textit{or} R \neq \nil \\[-.5ex]
&\wf(P \cdot ~\nil) = \wf(P) \\[-.5ex]
&\wf(\nil) = True.
\end{align*}
\end{small}

\vspace{-4ex}
\begin{remark}
	Note that the well-formedness property implies that if a process begins with a deterministic choice
	action \textit{if T then Q else R}, then all variables in $T$ must be instantiated, and thus only one branch may be taken.  
	For this reason, it is undesirable to specify processes that begin with such an action.
	Furthermore, note that the well-formedness property  avoids 
	explicit choices where both possibilities
	are the $\nil$ process. That is, processes containing either \textit{(if T then nil else nil)}, or \textit{(nil ? nil)}, respectively.
\end{remark}

We illustrate the notion of well-formed process
 below.

\begin{example}\label{ex:Det-tlslike}
 The behavior of a  Client initiating an instance of the handshake protocol from  Example~\ref{ex:tlsinf}  with the Server,
where the Server may or may not request the Client  to authenticate itself, may be
 specified  by the well-formed process shown below:

\vspace{-3ex}
   \begin{small}
 \begin{align*}
 (\mathit{Client}) \ &   +(hs ; n(\chV{C},r_1) ; \chV{G} ; gen(\chV{G}) ; keyG(\chV{G},\chV{C},r_2))\ \cdot \\[-.5ex]
   &  -(hs ; N ; \chV{G} ; gen(\chV{G}) ; E ; Z(\mathit{AReq},\chV{G},E,\chV{C},r_1,S,HM))\ \cdot \\[-.5ex]
   & \mathit{if}  ~( \mathit{AReq} = \mathit{authreq} ) \\[-.5ex]
     &  \mathit{then}  \\[-.8ex]
     & +(e(keyE (\chV{G},E,\chV{C},r_1), \\[-.5ex]
    &   ~~~~~~~~ sig(C, W(HM,\mathit{AReq},\chV{S},G_?,E,\chV{C},r_1)) ; \\[-.5ex]
   &  ~~~~~~~~ mac(keyE(\chV{G},E,,\chV{C},r_1), 
W(HM, \mathit{AReq},S,\chV{G},E,\chV{C},r_1))))\ \cdot \\[-.5ex]
  & \mathit{else}  \\[-.5ex] 
 &  +(e(keyE( \chV{G},E,\chV{C},r_2),  \\[-.5ex]
 & ~~~~~~~~  mac(keyE(\chV{G},E,\chV{C},r_2),  
W(HM, \mathit{AReq},S,\chV{G},E,\chV{C},r_1))))
 \end{align*}
 \end{small}


\vspace{-3ex}

 \noindent  where  
 $\mathit{KeyG}$, $Z$ and $W$ are macros used to construct
 messages sent in the protocol.  The variables
 $\chV{C}$ and   $\chV{G}$ are choice variables denoting
 the client and Diffie-Hellman group respectively, and the variables $r_1$ and $r_2$ are fresh variables.  All other variables are pattern variables.
 In particular, the variable $\mathit{AReq}$ is a pattern variable
 that  can be instantiated to either
 $\mathit{authreq}$ or $\mathit{noauthreq}$.  The Client makes a deterministic choice whether or not
 to sign its next message with its digital signature, depending
 on which value of $\mathit{AReq}$ it receives.
 
 \end{example}

\begin{example}\label{ex:non-tlslike}
The behavior of a Server who may or may not request a retry from a Client in an instance of the handshake protocol from  Example~\ref{ex:tlsinf}
may be specified as follows:

\vspace{-3ex}
 \begin{small}
 \begin{align*}
 (\mathit{Server}): \ &   -(hs ; N ; G ; gen(G) ; E) \cdot \\[-.5ex]
   &   (((+(hs ; retry) \cdot \\[-.5ex]
   & ~~~ -(hs ; N' ; G' ; gen(G') ; E') \cdot  \\[-.5ex]
& ~~~   +(hs ; n(\chV{S},r1) ; G' ; gen(G') ; keyG(G',\chV{S},r_2) ;  
Z(\chV{AReq},G',E',S,r_2,\chV{S},HM)))  \\[-.5ex]
&   ?   \\[-.5ex]
  &  ( + (hs ; n(\chV{S},r1) ; G ; gen(G) ; keyG(G,S,r_2) ; 
Z(\chV{AReq},G,E,S,r_2,\chV{S},HM)))))
 \end{align*}
 \end{small}
 \end{example}
 In this case the server nondeterministically chooses to request or not to
 request a retry.  In the case of a retry it waits for the retry message from the client, and then
 proceeds with the handshake message using the new key information from the
 client.  In the case when it does not request a retry, it sends the handshake message immediately
 after receiving the client's Hello message. The variable $r_2$ is a fresh variable, while $\chV{S}$ and $\chV{AReq}$
 are choice variables. $\chV{S}$ denotes the name of the server, and $\chV{AReq}$ is nondeterministically
 instantiated 
 to $\mathit{authreq}$ or $\mathit{noauthreq}$.

 
 \begin{example}\label{ex:ill-formed}
The following process 
does not satisfy the well-formedness property.

\vspace{-3ex}

\begin{small}
  \begin{align*}
   \mathit{(Resp)} &-(pk(B, A ; NA)) \cdot \\
					   & \ (+(pk(A, 1 ; n(B,r))) \ ? \ +(pk(A, 2))) \ \cdot \\
					   &  +(pk(\chV{C}, n(B,r) ))
  \end{align*}
  \end{small}

\vspace{-3ex}

  \noindent
  The problem with this process is the fresh variable $r$
  appearing in  message \allowbreak $+(pk(\chV{C}, n(B,r) ))$, since

  \begin{small}
  $$r \notin \bvar{ -(pk(B, A ; NA)) 
  \cdot (+(pk(A, 1 ; n(B,r))) \ ? \  +(pk(A, 2)))}$$
  \end{small}

\noindent
more specifically, because it does not appear in message   
$+(pk(A, 2))$, 
  but 
%
  $r \in \var{-(pk(B, A ; NA)) \cdot (+(pk(A, 1 ;  n(B,r))) \ ? \ +(pk(A, 2)))}.$
 \end{example} 
 
\subsection{Protocol Specification in Process Algebra} \label{sec:PASpecification}
Given a protocol $\caP$, we define  its specification in the 
protocol process algebra, written $\SpecPA$,
as a pair of the form $\SpecPA = ((\PASymbols,\PAPEq), \ProcPA)$,
where $(\PASymbols,\PAPEq)$ is an equational theory explained below, and
 $\ProcPA$ is a term denoting a \emph{well-formed} process configuration representing
the behavior of the honest principals as well as the 
 capabilities of the attacker. That is, 
$\ProcPA = (\mathit{ro_1}) P_1 ~\&~ \ldots  ~\&~ (\mathit{ro_i}) P_i $,
where each $ro_k$, $1 \leq k \leq i$, 
is either the role of an honest principal or
identifies one of the capabilities of the attacker.
 $\ProcPA$ cannot contain two processes with the same label, 
i.e., the behavior of each honest principal, and each attacker capability
are represented by a \emph{unique} process. $\PAPEq =\PEq \cup \PAEq$ is a set of equations with $\PEq$ denoting the protocol's cryptographic properties and $\PAEq$ denoting the properties of process constructors. 
The set of equations $\PEq$ is user-definable and can
vary for different protocols.   Instead,
the set of equations $\PAEq$ is always the same for all protocols. 
$\PASymbols=\PSymbols \cup \BNFSymbols$ is the signature
defining the sorts and function symbols as follows:
 
 \begin{itemize}
 \item $\PSymbols$ is an order-sorted signature defining the sorts and function symbols
        for the messages that can be exchanged in protocol $\caP$.
       However, independently of protocol $\caP$, $\PSymbols$  must always
        have  a sort  \sort{Msg} as the top sort in one of its 
       connected components.
        We call a sort $\sort{S}$ a \emph{data sort} iff it is either a subsort of $\sort{Msg}$, or there is a message constructor $c : \sort{S_1}...\sort{S}...\sort{S_n} \rightarrow \sort{S'}$, with $\sort{S'}$ a subsort of $\sort{Msg}$.
      The specific sort  \sort{Fresh} for fresh variables is an example of \emph{data sort}.
      Choice and pattern variables have sort \sort{Msg} or any
      of its subsorts.

 \item $\BNFSymbols$ is an order-sorted signature defining the sorts 
        and function symbols of the \emph{process algebra infrastructure}. 
        $\BNFSymbols$ corresponds
       exactly to the BNF definition 
       of the protocol process algebra's syntax in Section \ref{sec:syntaxPA}. 
        Although it has a sort \sort{Msg} for messages, it leaves this sort totally unspecified, 
        so that different protocols $\caP$ may use completely different message constructors and 
        may satisfy different equational properties $\PEq$. Therefore, $\BNFSymbols$
        will be the same signature for any protocol
       specified in the process algebra. 
       More specifically,  $\BNFSymbols$ 
       contains the sorts for  process configurations (\sort{ProcConf}),
        labeled processes (\sort{LProc}),  processes (\sort{Proc}), 
        constraints (\sort{Cond}), and messages(\sort{Msg}),
        as well as the subsort relations
        $\sort{LProc} \, < \, \sort{ProcConf}$.
       Furthermore, the function symbols in $\BNFSymbols$ are also defined according to the BNF definition.
\end{itemize}
Therefore, the syntax $\PASymbols$ of processes for $\caP$ will be in the union signature $\BNFSymbols \cup \PSymbols$,
consisting of the protocol-specific syntax $\PSymbols$, and the generic process syntax $\BNFSymbols$ through the shared sort \sort{Msg}.

\subsection{Operational Semantics of the Protocol Process Algebra}\label{sec:semanticsPA}

Given a protocol $\mathcal{P}$, 
a \emph{state} of $\caP$ consists of a set 
of (possibly partially executed) \emph{labeled processes}, and a set of terms in the intruder's knowledge $\{IK\}$.
That is, a state is a term of the form 
$\{ LP_1 \,\&\, \cdots \,\&\, LP_n ~|~ \{\textit{IK}\}\}$.
%
Given a state $St$ of this form,  we abuse notation and write 
$LP_k \in St$ if $LP_k$ is a labeled process in the set $LP_1 \,\&\, \cdots \,\&\, LP_n$.

The intruder knowledge $IK$ models the \emph{single} channel
through which all messages are sent and received. 
We consider an active attacker who has complete control of the channel, 
i.e, can read, alter, redirect, and delete traffic as well as create
its own messages by means of
\emph{intruder processes}.  That is, the purpose of some $LP_k \in St$
is to perform message-manipulation actions for the intruder.

State changes are defined by a set $\PARls$ of \emph{rewrite rules}, such that
the rewrite theory $(\PAStateSymbols, \PAPEq , \PARls)$ characterizes the behavior of 
protocol $\caP$, where  $\PAStateSymbols$ extends $\PASymbols$
by adding state constructor symbols.
We assume that a protocol's
execution begins with an empty state, i.e., a state with an empty
set of labeled processes, and an empty intruder knowledge. 
That is, the initial state is always of the form 
$ \{ \emptyset ~|~ \{empty\}  \}$.
Each transition rule in $\PARls$ is labeled with a tuple of the form $\mathit{(ro,i, j, a,n)}$, where:

\begin{itemize}
	\item $\mathit{ro}$ is  the role of the labeled process being executed in the transition.
     
    \item $i$ denotes the identifier of the labeled process being executed in the transition. Since there can be more than one process instance of the same role in a process state, $i$ is used to distinguish different instances, i.e., $ro$ and $i$ together uniquely identify a process in a state.
	
	\item $j$ denotes the process' step number since its beginning.

	\item $a$ is a ground term identifying the action that is being performed in the transition.
		  It has different possible values: 
		    ``$+m$'' or ``$-m$'' if the message $m$ was sent  (and added to the intruder's knowledge) or received, respectively;
		    ``$m$'' if the message $m$ was sent  but did not increase the intruder's knowledge,
		    ``$?$'' if the transition performs an explicit non-deterministic choice, or ``$\mathit{T}$'' if the transition 
		  performs a explicit deterministic choice.
	\item $n$ is a number that, if the action that is being executed is an explicit choice, indicates which branch has been chosen as the process continuation.
		  In this case $n$ takes the value of either $1$ or $2$.
	      If the transition does not perform any explicit choice, then $n=0$.
\end{itemize}


%
  
%
 
Below we describe the set of transition rules that define a protocol's
execution in the protocol process algebra, that is, the set of rules 
$\PARls$. 
Note that in the transition rules shown below, $PS$ denotes the rest of labeled processes of the state (which can be 
the empty set $\emptyset$).

\begin{itemize}
\item 
The action of \emph{sending a message} is represented by the two transition rules
below. Since we assume the intruder has complete control of the network, 
it can learn any message sent by other principals.
Rule~\eqref{eq:pa-output-modIK} denotes the case in which the
sent message is added  to the intruder's knowledge.
Note that this rule can only be applied if the intruder has not
already learnt that message.
%
Rule~\eqref{eq:pa-output-noModIK}
denotes the case in which the intruder chooses not to learn the message, i.e., the intruder's knowledge is not modified,
and, thus, no condition needs to be checked. 
Since choice variables denote messages that are nondeterministically chosen, all (possibly infinitely many) admissible ground substitutions for the choice variables are possible behaviors.

\begin{small}
 \begin{align}
&\{ (ro,i, j)~( +M \cdot P) ~\&~ PS  ~|~ \{IK\}  \} \notag\\[-.5ex]
&  \longrightarrow_{(ro,i,j,+M\sigma,0)}
\{ (ro,i, j+1)~P \sigma ~\&~ PS ~|~ \{ \inI{M\sigma}, IK\}  \} \notag\\[-.5ex]
&  \textit{ if } ( \inI{M\sigma}) \notin \textit{IK}  \notag\\[-.5ex]
&  \textit{where } \sigma  \ 
 \textit{is a ground substitution}
 \textit{ binding choice variables}  
\textit{ in} \  M   \eqname{{\small PA++}}
 \label{eq:pa-output-modIK}
  \end{align}

  
  \begin{align}
  &\{ (ro,i, j)~( +M \cdot P) ~\&~ PS  ~|~ \{IK\}  \} \notag\\[-.5ex]
  &  \longrightarrow_{(ro, i, j, M\sigma,0)}
 \{ (ro,i, j+1)~P\sigma ~\&~ PS ~|~ \{IK\} \} \notag\\[-.5ex]
 &  \textit{where } \sigma 
 \ \textit{is a ground substitution} 
\  \textit{binding choice variables} 
\textit{ in} \  M  \eqname{PA+}
 \label{eq:pa-output-noModIK}
 \end{align}
\end{small}

%
%
\item As shown in the rule below, a process can \emph{receive a message} matching a pattern $M$ if
there  is a message $M'$ in the intruder's knowledge, i.e., a message
previously sent either by some honest principal or by some intruder process, 
that matches the  pattern message $M$. 
After receiving this message the process
will continue with its variables instantiated by the matching substitution, 
which takes place modulo the equations $E_\caP$. Note that the intruder can ``delete" a message   via choosing not to learn it (executing Rule \ref{eq:pa-output-noModIK} instead of Rule \ref{eq:pa-output-modIK})  or not
to deliver it (failing to execute Rule \ref{eq:pa-input}).

\begin{small}
\begin{align}
&\{	(ro,i, j)~( -M\cdot P) ~\&~ PS \mid \{\inI{M'}, IK\} \}\notag\\[-.5ex]
&	\longrightarrow_{(ro,i,j,-M\sigma,0)}
\{	(ro,i, j+1)~P\sigma ~\&~ PS \mid \{ \inI{M'}, IK\}  \} \notag\\[-.5ex]
&~\textit{if}~ M'=_{\PEq} M\sigma  \eqname{PA-}
	\label{eq:pa-input}
\end{align}
\end{small}

\item The two transition rules shown below define the operational semantics of 
\emph{explicit deterministic choice}s. That is, the 
 operational semantics of an
  $\textit{if} ~T~ \textit{then} ~P~ \allowbreak \textit{else} ~Q$ expression.
More specifically, rule~\eqref{eq:pa-detBranch1} describes  the
\textit{then} case, i.e., if the constraint $T$ is satisfied, 
the process will continue as $P$.
Rule~\eqref{eq:pa-detBranch2} describes the \textit{else} case, 
that is, if the constraint $T$ is \emph{not} satisfied, the process will continue as $Q$.
Note that, since we only consider well-formed processes, these transition
rules will only be applied if $j \ge 1$. Note also that since $T$ has been fully substituted by the time the if-then-else is executed, and the constraints that we considered in this paper are of the form $ m \neq_{\PEq} m' $ or $m=_{\PEq} m'$, the satisfiability of $T$ can be checked by checking whether the corresponding ground equality or disequality holds.

\begin{small}
\begin{align}
 & \{ (ro,i, j)~((\textit{if} ~T~ \textit{then} ~P~ \textit{else} ~Q) ~\cdot R) ~\&~ PS \mid \{IK\}\}\notag\\[-.5ex]
 & \longrightarrow_{(ro,i, j, T,1)}
\{ (ro,i, j+1)~(P\cdot R) ~\&~ PS \mid \{IK\} \}
 ~~if~T  
 \eqname{PAif1}
 \label{eq:pa-detBranch1}
\end{align}


\begin{align}
& \{ (ro,i, j)~( (\textit{if} ~T~ \textit{then} ~P~ \textit{else} ~Q) ~\cdot R) ~\&~ PS \mid \{IK\} \} \notag\\[-.5ex]
& \longrightarrow_{(ro,i,j,T,2)}  
\{ (ro,i, j+1)~(Q\cdot R) ~\&~ PS \mid \{IK\} \}
 ~~if~  \neg T 
 \eqname{PAif2}
  \label{eq:pa-detBranch2}
\end{align}
\end{small}

\item The two transition rules below define the semantics of 
\emph{explicit non-deterministic choice} $P~?~Q$. In this case, 
the process can continue either as $P$, denoted by   rule~\eqref{eq:pa-nonDetBranch1},
or as $Q$, denoted by  rule~\eqref{eq:pa-nonDetBranch2}. 
Note that this decision is made non-deterministically. 

\begin{small}
\begin{align}
& \{ (ro,i, j)~((P~?~Q)\cdot R) ~\&~ PS \mid \{IK\} \}\notag\\[-.5ex]
&\longrightarrow_{(ro,i,j,?,1)} 
\{ (ro,i, j+1)~(P \cdot R) ~\&~ PS \mid \{IK\}  \} 
\eqname{PA?1}
\label{eq:pa-nonDetBranch1}\\[.5ex]
& \{ (ro,i, j)~((P~?~Q)\cdot R) ~\&~ PS \mid \{IK\} \}\notag\\[-.5ex]
&\longrightarrow_{(ro,i,j,?,2)} 
\{ (ro,i,j+1)(Q \cdot R) ~\&~ PS \mid \{IK\} \} 
\eqname{PA?2}
\label{eq:pa-nonDetBranch2}
\end{align}
\end{small}

\item 
The transition rules shown below describe the \emph{introduction of a new 
process} from the specification into the state, 
which allows us to support an unbounded session 
model. 
Recall that fresh variables are associated with a role and an identifier. 
Therefore, whenever a new process is introduced: 
(a) the largest process identifier $(i)$ will be increased by 1, and 
(b) new names will be assigned to the fresh variables in the new process.
The function $\MaxProcId(PS, ro)$ in the transition rule below is used to get the largest process identifier $(i)$ of role $ro$ in the process configuration $PS$.
The substitution $\rho_{ro, i+1}$ in the transition rule below takes a labeled process and assigns new names to the fresh variables according to the label.
More specifically, $(ro,{i+1},1)~ P_k(r_1, \ldots, r_n)\rho_{ro, i+1} = (ro,i+1,1)~ P_k(r_1, \ldots, r_n) \{r_1 \mapsto r_1.ro.i+1, \ldots, r_n \mapsto r_n.ro.i+1 \}$.         
 In a process state, a role name together with an identifier uniquely identifies a process. 
Therefore, there is a unique subset of fresh names for each process in the state. In the rest of this paper we will refer to this kind of substitutions as \emph{fresh substitutions}.  

\vspace{-1.5ex}
\begin{small}
 \begin{align}
 \left \{
 \begin{array}{@{}l@{}}
  \forall \  (ro)~ P_k \in \ProcPA\notag\\[.5ex]
  \{ PS \mid \{IK\} \} \notag\ \\[-.5ex]
  \longrightarrow_{(ro, i+1, 1,A,Num)}
 \{ (ro, i+1, 2)~P'_k ~\&~ PS \mid \{IK'\}  \}\notag\\ [1ex]
  \textsf{IF} ~  \{(ro,i+1,1)~ P_k\rho_{ro,i+1} \mid \{IK\}  \} \notag\\[-.5ex]
\longrightarrow_{(ro,i+1, 1,A,Num)} 
 \{	(ro,i+1, 2)~P'_k ~ \mid \{IK'\} \} \\[1ex]
  \textit{where } \rho_\mathit{ro,i+1} 
 \ \textit{is a fresh substitution}, \notag\\
 \  i= \MaxProcId(PS, ro)  
 \end{array}
 \right \}\eqname{PA\&}
 \label{eq:pa-new}
 \end{align}
 \end{small}
 \vspace{-1.5ex}

 \noindent
 Note that $A$ denotes the action of the state transition,
 and can be of any of the forms explained above.
 The function $\MaxProcId$ is defined as follows:

\vspace{-1.5ex}
\begin{small}
\begin{align*}
& \MaxProcId(\emptyset, ro) = 0 \notag \\
& \MaxProcId((ro, i, j) P \& PS, ro)  = max(\MaxProcId(PS, ro), i) \notag \\
& \MaxProcId((ro', i, j) P \& PS, ro) = \MaxProcId(PS, ro) 
\hspace*{10mm} \textit{if} \ ro \neq ro' \notag 
\end{align*}
\end{small}
\vspace{-1.5ex}

\noindent
where $PS$ denotes a process configuration, $P$ denotes a process, and $ro, ro'$ denote role names.
\end{itemize} 
 
 Therefore, the behavior of a protocol in the process algebra is
 defined by the set of transition rules
 $\PARls = \{ \eqref{eq:pa-output-modIK} ,  \ \eqref{eq:pa-output-noModIK} , 
              \eqref{eq:pa-input} ,
              \eqref{eq:pa-detBranch1} ,    \eqref{eq:pa-detBranch2} , \
              \eqref{eq:pa-nonDetBranch1} , \ \allowbreak  \eqref{eq:pa-nonDetBranch2} \} \cup
              \eqref{eq:pa-new} $.
              
Our main result is a bisimulation between the state space generated by 
the transition rules $\PBWRls$, 
associated to the symbolic backwards semantics of Section~\ref{sec:mnpa-syntax},
and 
the transition rules $\PARls$ above, associated to the forwards semantics for process algebra.
This is nontrivial, since there are three major
 ways in which the two semantics differ.  
 The first is that  processes  ``forget'' their past,  while strands ``remember'' theirs.  
 The second is that Maude-NPA uses backwards search, while the process
 algebra proceeds forwards.  
 The third is that Maude-NPA performs symbolic reachability analysis using
 terms with variables, while the process algebra considers only ground terms. 

 We systematically relate these different semantics
by introducing an intermediate semantics, a forward strand space
semantics extending that  in \cite{Escobar2014FW}. 
First, in Section~\ref{sec:Cstr} we extend the strand space model with constraints, since strands are the basis
of both the forwards semantics and the backwards semantics of Maude-NPA.
In Section~\ref{sec:CstrFW} we 
augment
the forwards strand space semantics of \cite{Escobar2014FW} 
with choice operators and operational semantic rules to produce a \emph{constrained forwards semantics}.
In Section~\ref{sec:PA-CstrFW-bisim} we prove bisimilarity of
the process algebra semantics of Section \ref{sec:PA}
and the constrained forwards semantics of Section \ref{sec:CstrFW}.
In~\cite{Escobar2014FW} 
the forwards strand space semantics
was proved sound and complete w.r.t. the original symbolic backwards semantics of Maude-NPA.
But now such proofs had to be extended to handle  constraints.
In Section~\ref{sec:CstrBW} we also
augment
the original symbolic backwards semantics of Maude-NPA
with choice operators and operational semantic rules to produce a
\emph{constrained backwards semantics}.
In Section~\ref{sec:proof-CstrFW-CstrBW},
we then prove  
that the constrained backwards semantics is sound and complete
with respect to the constrained forwards semantics.
By combining the bisimulation
between the process algebra and the constrained forwards semantics on the one hand,
and the bisimulation 
between the constrained forwards semantics and the constrained backwards semantics on the other hand,
we obtain the main bisimulation result.

Besides providing a detailed semantic account of how the strand model
can be extended with choice features, the key practical importance of
these bisimulation results is that, with the relatively modest
extensions to Maude-NPA described in Section \ref{Integration}
and supported by its recent 3.0 release, sound and complete analysis
of protocols with choice features specified in process algebra is made possible.

\section{Constrained Protocol Strands with Choice} \label{sec:Cstr}
To specify and analyze protocols with choices in Maude-NPA, in this section we extend Maude-NPA's strand notation
by adding new symbols
to support explicit choices. 
We refer to the strands in this extended syntax as \emph{constrained protocol strands}.

In Section~\ref{sec:syntaxCstrFW}
we  describe the syntax for constrained protocol strands.
Then, in Section~\ref{sec:mappingToCstrFW} we define a 
mapping from a protocol specification in the protocol  process algebra, 
as described in Section ~\ref{sec:PASpecification}, 
to a specification based on constrained protocol strands. 

 \subsection{Constrained Protocol Strands Syntax}\label{sec:syntaxCstrFW}

  In this section we extend  Maude-NPA's syntax by adding \emph{constrained messages}, which
  are terms of the form $\{\mathit{Cstr}, \allowbreak \mathit{Num}\}$,
  where $\mathit{Cstr}$ is a constraint,
   and $\mathit{Num}$ is a natural number that identifies
   the continuation of the protocol's execution, among the 
   two possibilities after an explicit choice point. 
  More specifically, we extend the $\SSSymbols$ part of the signature $\NPASymbols$ of the Maude-NPA's syntax we defined in Section \ref{sec:mnpa-syntax} as follows:

  \begin{itemize}
  	\item A new sort \sort{Cstr} represents the constraints 
  	      allowed in constrained messages,
	      containing three symbols:
(i)~$? \ : \ \to \sort{Cstr}$,
(ii)         	$\_{=}\_ \ : \ \sort{Msg} \ \sort{Msg} \to \sort{Cstr}$,
and 
(iii)            $ \_{\neq}\_ \ : \ \sort{Msg} \ \sort{Msg} \to \sort{Cstr}$.
  	        	      
      	
  	\item A new sort \sort{CstrMsg} for constrained messages, such that \sort{CstrMsg} $<$ \sort{SMsg}, where \sort{SMsg} is an existing Maude-NPA sort denoting signed messages (i.e., messages with + or -).
  	Therefore, now a strand is a sequence of output, input 
  	and constrained messages.
  		
  	\item A new operator 
$\{\_,\_\} : \sort{Cstr} \ \sort{Nat} \to \sort{CstrMsg} \ $
	constructs constrained messages.
  \end{itemize}

We refer to this extended signature as $\PCstrSymbols$.   
Note that the protocol signature $\Sigma_\caP$ is contained in
$\NPASymbols$,  and therefore in $\PCstrSymbols$.
Furthermore, in the constrained semantics we allow each honest principal or intruder capability strand to be 
\emph{labeled} by the ``role'' of that strand in the protocol 
(e.g., \texttt{(Client)} or \texttt{(Server)}).
  Therefore,  strands are now terms
of the form 
$(ro, i)[ u_1, \ldots, u_n]$,  
where $ro$ denotes the role of the strand in the protocol, 
$i$ is a unique identifier distinguishing different instances of strands of the same role, 
and each $u_i$ can be
  a sent or received message,
   i.e., a term of the form $M^\pm$, or a constraint message of the form $\{\mathit{Cstr},Num\}$. 
We often omit $i$, or both $ro$ and $i$ for clarity when they are not relevant.

\subsection{Protocol Specification using Constrained Protocol Strands}\label{sec:mappingToCstrFW}

The behavior of a  protocol involving choices can be specified 
using the syntax presented in Section~\ref{sec:syntaxCstrFW} as described  below.

\begin{definition}[Constrained protocol strand specification]\label{def:specCstrFW}
	Given a 
	protocol $\caP$, we define  its specification by means of 
	constrained protocol strands, written $\SpecCstrFW$,
	as a tuple of the form $\SpecCstrFW = ((\PCstrSymbols,\CstrEq),  \allowbreak \StrCstrFW)$,
	where $\PCstrSymbols$ is the protocol's signature (see Section~\ref{sec:syntaxCstrFW}), 
	and
	$\CstrEq=\PEq \cup \SSEq $ is a set of equations as we defined in Section~\ref{sec:mnpa-syntax},
	where $\PEq$ denotes the protocol's cryptographic properties and $\SSEq$ denotes  the protocol-independent properties of constructors of strands. That is, the set of equations $\PEq$ may vary depending on different protocols, but the set of equations $\SSEq$ is always the same for all protocols.
	 $\StrCstrFW$ is a set of constrained protocol strands as defined  in Section~\ref{sec:syntaxCstrFW},  representing the behavior of the honest principals as well as the 
	capabilities of the attacker. That is, 
	$\StrCstrFW$ is a set of labeled strands of the form:
	$\StrCstrFW = \{ (\mathit{ro_1}) [u_{1,1}, \ldots, u_{1,n_1}] ~\&~  \allowbreak \ldots  ~\&~  
	(\mathit{ro_m}) [u_{m,1}, \allowbreak \ldots,  u_{m,n_m}]  \}$,
	where, for each $ro_k$ such that $1 \leq k \leq i$, 
	$ro_k$ is either the role of an honest principal, or
	identifies one of the capabilities of the attacker.
	We note that $\StrCstrFW$  \emph{may contain several strands with the same
	label}, each defining one of the possible paths of such a principal.
\end{definition}

The protocol specification described above can be obtained 
by \emph{transforming} a specification in the process algebra of Section~\ref{sec:PASpecification}
as follows. 
Given a protocol $\caP$, its specification in the process algebra $\ProcPA$, consists of a set
of \emph{well-formed} labeled processes.
We transform a term denoting a set of labeled processes
into a term denoting a set of constrained protocol strands
by the mapping $\toCstrFW$. The intuitive idea is that, since our process contains no recursion, each process can be  
``deconstructed" as a set of constrained protocol strands, where each such strand represent a possible execution path of the process.

The mapping $\toCstrFW$ is specified in Definition~\ref{def:transf} below.
 
\begin{definition}[Mapping labeled processes $\toCstrFW$]\label{def:transf}
Given a labeled process $\mathit{LP}$
and a process configuration $\mathit{LPS}$, 
we define the mapping 
 $\toCstrFW : \caT_{\PASymbols}(\Variables) \ \to \caT_{\PCstrSymbols}(\Variables)$
recursively as follows:

\vspace{-1ex}
\begin{small}
\begin{align*}
& \toCstrFW(\mathit{LP} \ \& \ \mathit{LPS}) = \toCstrFWa(  \mathit{LP}, \nil) \ \& \ \toCstrFW(  \mathit{LPS})   \\[-.5ex]
& \toCstrFW(\emptyset) =  \mathit{\emptyset} 
\end{align*}
\end{small}
\vspace{-1.5ex}

\noindent
where
 $\emptyset$ is the empty set of strands.
 $\toCstrFWa$ is an auxiliary mapping that maps a term denoting a labeled process to a term that denotes a set of constrained protocol strands.
 It takes two arguments: a labeled process, and a temporary store that keeps a sequence of messages. 
More specifically, 
 $\toCstrFWa : \caT_{\PASymbols}(\Variables) \ \times  \caT_{\PCstrSymbols}(\Variables) \to \caT_{\PCstrSymbols}(\Variables)$ is defined as follows: 

 \vspace{-1ex}
\begin{small}
\begin{align*}
&  \toCstrFWa( (ro,i, j)  \ \nil, L) = (ro, i) \  [ \ L  \  ] \\[-.5ex]
&  \toCstrFWa( (ro,i, j) \ +M \ . \  P, L) 
=   \toCstrFWa( (ro,i, j) \  P, (L, \  +M) )   \\[-.5ex]
&  \toCstrFWa( (ro,i, j) \  -M \  . \  P, \  L) 
=   \toCstrFWa( (ro,i, j) \   P, (L, \  -M)  )   \\[-.5ex]
&  \toCstrFWa((ro,i, j) \  \textit{ (if } T  \textit{ then }  P \textit{ else } Q ) \  . \ R, L) \\[-.5ex]
& \hspace*{10mm} = \ \toCstrFWa((ro,i, j) \ P  \  .  \  R, (L,  \  \{T,1\}) ) \ \&\ 
\toCstrFWa((ro,i, j) \ Q  \ .  \  R, (L,  \  \{\neg T, 2\})) \\[-.5ex]
&  \toCstrFWa( (ro,i, j)\ (P \ ? \ Q)\ . \ R, L) \\[-.5ex]
&  \hspace*{10mm}= \  \toCstrFWa( (ro,i, j) \ P \ . \ R, (L,\  \{?,1\})) \ \&\ 
\toCstrFWa( (ro,i, j) \ Q\  . \ R, (L, \{?,2\})) 
\end{align*}
\end{small}
\vspace{-1ex}

\noindent
where $P$, $Q$, and $R$ denote processes,
$M$ is a message, 
$T$ is a constraint,
and $L$ denotes a list of messages, i.e., 
 input, output or constraint messages.
 
Note that $\toCstrFW$ does not modify output and input messages,
since messages are actually terms in 
$\caT_{\PSymbols/\PEq}(\Variables)$
in both the protocol process algebra,
and the constrained forwards semantics.
$\toCstrFW$ can be used both as a map between specifications, 
and as a map from process configurations and strand sets appearing in states.
\end{definition}

We illustrate the $\toCstrFW$ transformation  with the example below.

\begin{example}\label{ex:toCstrFW}
If we apply the mapping $\toCstrFW$ to the process in Example~\ref{ex:non-tlslike}
we obtain the following term which denotes a 
set of strands:

\vspace{-2ex}
\begin{small}
\begin{align*}
 (\textit{Server}) \ & [ \ \{?,1\}, \\[-.5ex]
 &  -(hs ; N ; G ; gen(G) ; E) , \\[-.5ex]
   &   +(hs ; retry) , \\[-.5ex]
   & -(hs ; N' ; G' ; gen(G') ; E')) ,  \\[-.5ex]
& +(hs ; n(\chV{S},r1) ; G' ; gen(G') ; keyG(G',\chV{S},r_2) ;  
Z(\chV{AReq},G',E',S,r_2,\chV{S},HM)) ] \ \& \\[-.5ex]
 (\textit{Server}) \ &[ \ \{?,2\}, \\[-.5ex]
 & -(hs ; N ; G ; gen(G) ; E) , \\[-.5ex]
   &   + (hs ; n(\chV{S},r1) ; G ; gen(G) ; keyG(G,S,r_2) ; 
Z(\chV{AReq},G,E,S,r_2,\chV{S},HM)) ]
 \end{align*}
 \end{small}
\vspace{-5ex}
 \end{example}
 

A protocol specification in the protocol process algebra
can then be transformed into a specification of that protocol in the constrained protocol strands described below using $\toCstrFW$.

\begin{definition}[Specification transformation]\label{def:transf-pa-Cstr}	
Given a protocol $\caP$ and its protocol process algebra specification
 $\SpecPA = ((\PASymbols,\PEq \cup \PAEq), \ProcPA)$,
with
$\ProcPA =  (ro_1) P_1  \& \ldots    \& 
   (ro_n) P_n $,
 its specification by means of constrained protocol strands  
is $\SpecCstrFW = ((\PCstrSymbols,\PEq \cup \SSEq), \allowbreak\StrCstrFW)$ with
  $\StrCstrFW = \toCstrFW(\ProcPA)$.
\end{definition}

\section{Constrained Forwards Strand Semantics}\label{sec:CstrFW}
In this section we extend Maude-NPA's rewriting-based
forwards semantics in \cite{Escobar2014FW} by adding new 
transition rules for constrained messages.
We refer to this extended forwards semantics
as \emph{constrained forwards strand semantics}. 
We show that the process algebra semantics
and the constrained forwards strand semantics are label bisimilar. 
Therefore, protocols exhibiting choices
can be specified and executed in an equivalent way in both semantics. 

In the constrained forwards strand semantics, state changes are defined by a set $\PCstrFWRls$ of \emph{rewrite rules}, so that
the rewrite theory $(\PCstrSymbols, \CstrEq, \PCstrFWRls)$ characterizes the behaviors of 
protocol $\caP$.

The set of transition rules  $\PCstrFWRls$ is an extension of the transition rules $\FWRls$ in \cite{Escobar2014FW}. 
The transition rules are generated from the protocol specification.
A \emph{state} consists of a multiset of partially executed strands and a set of terms denoting the intruder's knowledge.
The main differences between the sets $\PCstrFWRls$ and   $\FWRls$
are:
(i)  new transition rules are added in $\PCstrFWRls$ to appropriately deal with constraint messages,
(ii) strands are labeled with the role name, together with the identifier for distinguishing different instances, as explained in Section~\ref{sec:syntaxCstrFW}, 
(iii) transitions are also labeled, similarly as in the protocol process algebra, 
(iv) the global counter for generating fresh variables is deleted from
the state.
Instead, special unique names are assigned to fresh variable, which simplifies our notation.
 

In the constrained forwards strand semantics we label each transition rule 
similarly as in Section~\ref{sec:semanticsPA}, that is, using labels
of the form $(ro,i,j,a,n)$, where $ro$, $i$, $a$, and $n$ are as explained in  
Section~\ref{sec:semanticsPA}, and $j$ in this case 
is the position of the message that is being exchanged in the state transition. 
Also, similar to Section~\ref{sec:semanticsPA}, for transitions that send out messages containing choice variables, all (possibly infinitely many) admissible ground substitutions for the choice variables are possible behaviors.
A similar mechanism for distinguishing different fresh variables is used as that explained in Section~\ref{sec:semanticsPA}.
Since messages are introduced into strands in the state incrementally, 
we instantiate the fresh variables incrementally as well. Recall that fresh variables always first show up in a sent message.
Therefore, each time a sent message is introduced into a strand in the state, we assign new names to the fresh variables in the message being introduced.
 The function $\MaxStrId$ for getting the max identifier for a constrained strand of a certain role is similar to $\MaxProcId$ in Section \ref{sec:semanticsPA}. 

Since now messages in a strand can be sent or received messages, i.e., terms of 
the form $m^+$ or $m^-$, as well as constraint messages $\{\mathit{Cstr},\mathit{Num}\}$,
we represent them in the rules below simply as terms of the form $u_i$ when their
exact form is not relevant.  We will use the precise form of the message 
when disambiguation is needed.  

Before explaining the new transition rules for constraint messages, 
we  show how the transition rules in \cite{Escobar2014FW} are labeled.

The constrained forwards strand semantics extends Maude-NPA's forwards semantics in \cite{Escobar2014FW} by adding
transition rules to handle constraint messages, i.e,  messages of the form
$\{\mathit{Cstr},\mathit{Num}\}$, where $Num$ can be either $1$ or $2$.
First, we add the two transition rules below 
for the cases when  such a constrained message  
 comes from explicit choices. 
Note that, as a consequence of 
well-formedness, the constraints introduce no new variables, and since
the constraints that we 
consider are of the form $ m \neq_{\PEq} m' $ or $m=_{\PEq} m'$, the
satisfiability of $Cstr$ can 
be checked by checking whether the corresponding ground equality or disequality holds.

\begin{small}
\begin{align}
\left\{
\begin{array}{@{}l@{}}
\forall \ (ro) \  [ u_1,\ldots,u_{j-1},u_{j}^+,u_{j+1},\ldots,u_n] \in \StrCstrFW  \wedge  j \textgreater 1:  \\[1ex]
\{SS \,\&\,\{IK\} 
\,\&\, (ro, i) \ [ u_1,\ldots,u_{j-1}]\} \\
\hspace{5ex} 
\hspace{-4ex}\to_{(ro,i, j,(u_{j}\rho_\mathit{ro,i}\sigma)^+,0)}  \\ 
\{
SS 
\,\&\, \{\inI{u_{j}\rho_{ro,i}\sigma},IK\}
\,\&\, (ro, i) \ [ u_1,\ldots,u_{j-1},(u_{j}\rho_\mathit{ro,i}\sigma)^+ ]\} \\
\hspace{8ex} 
\ \textsf{IF}\   (\inI{u_j\rho_\mathit{ro,i}\sigma})    \notin IK \\[1ex]
 \  \textit{where }  \sigma 
 \  \textit{is a ground substitution} 
 \textit{ binding choice variables} 
  \ \textit{in} \  u_j , 
\\  
 \ \rho_\mathit{ro,i} =\{r_1\mapsto r_1.ro.i,  \ldots, r_n\mapsto r_n.ro.i  \}  
 \ \textit{is a fresh}  
  \ \textit{substitution}. 
\end{array}
\right\}  \eqname{F++}
\label{eq:forward-positive-modIK-Cstr}
\end{align}

\begin{align}
\left\{
\begin{array}{@{}l@{}}
\forall \ (ro) \ [ u_1,\ldots,u_{j-1},u_{j}^+,u_{j+1},\ldots,u_n] \in \StrCstrFW  \wedge  j \textgreater 1:  \\[1ex]
\{SS \,\&\,\{IK\}
\,\&\, (ro, i) \ [ u_1,\ldots,u_{j-1}]\} \\
\hspace{10ex} 
\hspace{-8ex}\to_{(ro,i,j,u_{j}\rho_\mathit{ro,i}\sigma,0)} \\ 
\{
SS 
\,\&\, \{IK\}
\,\&\,(ro, i) \ [ u_1,\ldots,u_{j-1},(u_{j}\rho_\mathit{ro,i}\sigma)^+ ]\}\\ [1ex]
 \ \textit{where }  \sigma 
 \ \textit{is a ground substitution} 
 \textit{ binding choice variables}  
 \ \textit{in} \  u_j , \\
 \ \rho_\mathit{ro,i} =\{r_1\mapsto r_1.ro.i,  \ldots, r_n\mapsto r_n.ro.i  \}  
 \ \textit{is a fresh}  
  \ \textit{substitution}. 
\end{array}
\right\}  \eqname{F+}
\label{eq:forward-positive-Cstr}
\end{align}

\noindent 
\begin{align}
\left\{
\begin{array}{@{}l@{}}
\forall \ (ro) \ [ u_{1}^+,\ldots,u_n] \in \StrCstrFW: \\[1ex]
\{SS \,\&\,\{IK\} \}
\to_{(ro, i+1, j,(u_{1}\rho_\mathit{ro,i+1}\sigma)^+,0)}\\[1ex]
\{
SS \,\&\,(ro, i+1) \ [ (u_{1}\rho_\mathit{ro,i+1}\sigma)^+ ]
\,\&\, \{\inI{u_{1}\rho_\mathit{ro,i+1}\sigma},IK\} \}
\\
\hspace{2ex}		
\hspace{2ex} \ \textsf{IF}\  (\inI{u_1\rho_\mathit{ro,i+1}\sigma})    \notin IK \\ [1ex]
 \ \textit{where }  \sigma
 \ \textit{is a ground substitution}  
 \textit{ binding choice variables} 
 \ \textit{in} \  u_1,\\  
 \  i=\MaxStrId(SS, ro), \\
\ \rho_\mathit{ro,i+1} =\{r_1\mapsto r_1.ro.i+1,  \ldots, r_n\mapsto r_n.ro.i+1  \}  
 \ \textit{is a fresh substitution}.
\end{array}
\right\}  \eqname{F++\&}
\label{eq:forward-positive2-modIK-Cstr}
\end{align}

\noindent
\begin{align}
\left\{
\begin{array}{@{}l@{}}
\forall \ (ro) \ [ u_{1}^+,\ldots,u_n] \in \StrCstrFW:\\[1ex]
\hspace{5ex}
\hspace{-5ex} \{SS \,\&\,\{IK\}\} 
\to_{(ro,i+1, j,u_{1}\rho_\mathit{ro,i+1}\sigma,0)}\\
\hspace{2ex}\{ SS \,\&\,(ro, i+1) \  [ (u_{1}\rho_\mathit{ro,i+1}\sigma)^+ ]
\,\&\, \{IK\}\} \\[1ex]
 \ \textit{where }  \sigma
 \ \textit{is a ground substitution}  
 \textit{ binding choice variables} 
\ \textit{in} \  u_1,  \\
\ i=\MaxStrId(SS, ro),\\
 \ \rho_\mathit{ro,i+1} =\{r_1\mapsto r_1.ro.i+1,  \ldots, r_n\mapsto r_n.ro.i +1 \}  
 \ \textit{is a} 
 \  \textit{fresh substitution}.
\end{array}
\right\}  \eqname{F+\&}
\label{eq:forward-positive2-Cstr}
\end{align}

\noindent
	\begin{align}
	\left\{
	\begin{array}{@{}l@{}}
	\forall \ (ro) \ [ u_1,\ldots,u_{j-1},u_{j}^-,u_{j+1},\ldots,u_n] \in \StrCstrFW 
	\wedge  j > 1: \\[1ex]
	\{SS \  \& \{\inI{u_{j}},IK\}
	\,\&\, (ro, i) \ [ u_1,\ldots,u_{j-1}]\}   \\
	\to_{(ro, i, j, u_j^-,0)}\\ 
	\{
	SS 
	\,\&\, \{\inI{u_{j}},IK\}
	\,\&\, (ro, i) \ [ u_1,\ldots,u_{j-1},u_{j}^- ]\}
	\end{array}
	\right\}   \eqname{F-}
	\label{eq:forward-negative-Cstr}
	\end{align}
	
	\noindent
	\begin{align}
	\left\{
	\begin{array}{@{}l@{}}
	\forall (ro) \ [ u_{1}^-,u_{2},\ldots,u_n] \in \StrCstrFW :  \\[1ex]
	\{SS \,\&\,\{\inI{u_{1}},IK\} \} \\
	\to_{(ro,i+1, 1, u_1^-,0)}
	\{
	SS \  \& \ (ro, i+1) \  [ u_1^- ] 
	\,\&\, \{\inI{u_{1}},IK\} \} \\
	\ \ \textit{where }  i=\MaxStrId(SS, ro)
	\end{array}
	\right\}   \eqname{F-\&}
	\label{eq:forward-negative2-Cstr}
	\end{align}
\end{small}

\begin{small}
    \begin{align}
    \left\{
    \begin{array}{@{}l@{}}
    \forall \ (ro) \ [ u_1,\ldots,u_{j-1},\{\mathit{Cstr},\mathit{Num}\},u_{j+1},\ldots,u_n] \in \StrCstrFW \\
    \wedge  j > 1: \\[1ex]
    \{SS \  \& \{IK\}
    \,\&\,(ro, i) \ [ u_1,\ldots,u_{j-1}]\}   \\
    \to_{(ro,i,j,T,\mathit{Num})}  \\ 
    \{
    SS 
    \,\&\, \{IK\}
    \,\&\,(ro, i) \ [ u_1,\ldots,u_{j-1},\{\mathit{Cstr},\mathit{Num}\} ]\}
    \\
    \hspace{2ex}		
    \hspace{2ex} \ \textsf{IF}\  \mathit{Cstr}
    \end{array}
    \right\}  \eqname{Fif}
    \label{eq:forward-CstrMsg-det}
    \end{align}
    
      \begin{align}
      \left\{
      \begin{array}{@{}l@{}}
      \forall \ (ro) \ [ u_1,\ldots,u_{j-1},\{?,\mathit{Num}\},u_{j+1},\ldots,u_n] \in \StrCstrFW \\
      \wedge  j > 1: \\[1ex]
      \{SS \  \& \{IK\}
      \,\&\,(ro, i) \ [ u_1,\ldots,u_{j-1}]\}   \\
       \to_{(ro,i,j,?,\mathit{Num})} \\ 
      \{
      SS 
      \,\&\, \{IK\}
      \,\&\,(ro, i) \ [ u_1,\ldots,u_{j-1},\{?,\mathit{Num}\} ]\}
      \end{array}
      \right\} \eqname{F?}
      \label{eq:forward-CstrMsg-nonDet}
      \end{align}
\end{small}
        
The following set of transition
 rules adds to the state a new strand whose first message is 
 a constraint message of the form $\{?,\mathit{Num}\}$:

	\noindent
	\begin{small}
   \begin{align}
 	\left\{
 	\begin{array}{@{}l@{}}
 		\forall \  (ro) \ [ \, \{?,\mathit{Num}\},u_{2},\ldots,u_n] \in \StrCstrFW :           \\[1ex]
 		\{SS \,\&\,\{IK\} \}           \\                                 
 		\to_{(ro,i+1, 1,?,\mathit{Num})}                                                    \\
 		\{
 	SS \  \& \ (ro, i+1) \  [ \, \{?,\mathit{Num}\} \, ] 
 	\,\&\, \{IK\} \} \\
 	\ \ \textit{where }  i=\MaxStrId(SS, ro)
 	\end{array}
 	\right\}  \eqname{F?\&}
 	\label{eq:forward-CstrMsg-first}
 	\end{align}	
	\end{small}

	\begin{definition}\label{def:CstrFW-semantics}
		Let $\mathcal{P}$ be a protocol with signature $\PCstrSymbols$ and equational theory $\CstrEq$.  We define the {\em constrained forwards rewrite theory characterizing $\mathcal{P}$} as 
		$(\PCstrSymbols, \CstrEq, \allowbreak \PCstrFWRls)$
		where $\PCstrFWRls = \eqref{eq:forward-positive-modIK-Cstr}
		\cup \eqref{eq:forward-positive-Cstr}
		\cup  \eqref{eq:forward-positive2-modIK-Cstr}
		\cup  \eqref{eq:forward-positive2-Cstr}
		\cup  \eqref{eq:forward-negative-Cstr} 
		\cup  \eqref{eq:forward-negative2-Cstr} 
		\cup \eqref{eq:forward-CstrMsg-det}
		\cup \eqref{eq:forward-CstrMsg-nonDet} 
		\cup \eqref{eq:forward-CstrMsg-first}$. 
	\end{definition}

\section{Bisimulation between Constrained Forwards Strand Semantics and Process Algebra Semantics}
\label{sec:PA-CstrFW-bisim}

In this section we show that the process algebra semantics
and the constrained forwards strand semantics are label bisimilar. We first define PA-State and FW-State, the respective notions
of state in each semantics.

\begin{definition}[PA-State]
 Given a protocol $\caP$, a \emph{PA-State} of $\caP$
 is a state in the protocol process algebra semantics
that is \emph{reachable} from the initial state. The initial PA-State is $P_{\mathit{init}}=\{\emptyset ~|~ \{empty\} \}$. 
\end{definition}

\begin{definition}[FW-State]
Given a protocol $\caP$, a \emph{FW-State} of $\caP$
is a state in the constrained forwards strand semantics
that is \emph{reachable} from the initial state. The initial FW-State is $F_{\mathit{init}}= \{\emptyset ~\&~ \{empty\}\} $. 
\end{definition}

The bisimulation relation is defined based on reachability, i.e., if a PA-State and a FW-State are in the relation $\HState$, then they 
both 
can be reached from their corresponding initial states by the same
label sequence. Note that we only consider states that are reachable
from the initial states. 

Let us first define the notation of label sequence that we will use throughout.

\begin{definition}[Label Sequence]
An ordered sequence $\alpha$ of transition labels is defined by using $\_.\_$ as an associative concatenation operator with $nil$ as an identity. The length of a label sequence $\alpha$ is denoted by $| \alpha |$. 
Given a label sequence $\alpha$, we denote by $\alpha|_{(ro, i)}$ the sub-sequence of labels in $\alpha$ that have $ro$ as role name, and $i$ as identifier, i.e., labels of the form $(ro, i, \_,\_,\_)$ ($\_$ is a shorthand for denoting any term). 
\end{definition}

  \begin{definition}[Relation $\HState$]
Given a protocol $\caP$, the relation $\HState$ is defined as:
$\HState=\{(\Pst, \Fst) \in \textit{PA-State} \times \textit{FW-State} ~|~ \exists \ \textit{label sequence} \ \alpha  \textit{ s.t. }  P_{init}\rightarrow_{\alpha} \Pst, \allowbreak ~F_{init} \rightarrow_{\alpha} \Fst \}$.
 \end{definition}

Recall that a process can be ``deconstructed" by the mapping
$\toCstrFW$ into a set of constrained protocol strands, each
representing a possible execution path. If a PA-State $\Pst$ and a
FW-State $\Fst$ are related by $\HState$, then an important
observation is that there is a  \emph{duality} between individual
processes in $\Pst$ and strands in $\Fst$: if there is a process in
the $\Pst$ describing a role's \emph{continuation} in the future,
there will be a corresponding strand in $\Fst$ describing the part of
the process that has \emph{already been executed}, and vice versa.
Another observation is that, since the intruder's knowledge is
extracted from the communication history, following the definition of
$\HState$, the states $\Pst$ and $\Fst$ have the same communication
history, therefore they have the same intruder's knowledge. We
formalize these observations in Lemmas \ref{lemm:equiv-PStr} and
\ref{lemm:sameIK}. These lemmas then lead us to the 
main result that $\HState$ is a bisimulation  relation.

We now 
define the relation $\HProc$, which relates a possibly partially executed labeled process and a constrained strand.  This relation defines the \emph{duality} relation between a labeled process and a constrained strands. 
If a labeled process $LP$ is related to a constrained strand $Str$ by the relation $\HProc$, then:
 (i) $LP$ and $Str$ denote the behavior of the  same role with the same identity in the same protocol, and 
 (ii) for any strand $Str_\mathit{LP}$, $Str_\mathit{LP}$ denotes a possible execution path of $LP$ iff $Str$ followed by $Str_\mathit{LP}$ forms a valid possible execution path of the protocol.  

\begin{definition}[Relation $\HProc$]
 Given a protocol $\caP$, and a possibly partially executed labeled process $LP$ of $\caP$, a possibly partially executed constrained strand $Str$ of $\caP$, then $(LP , Str)\in \HProc$ iff  
 
 \begin{quote}
 	 $\toCstrFW(LP) = 
 	  \&
 	  \{ (ro, i)  
 	  [u_{j+1},\ldots,  u_n]\rho_{ro,i}\theta \mid 
 	 \exists\ \textit{ground substitution}\ \theta \ \linebreak
 	 \exists (ro)[u_1,\ldots u_j, u_{j+1},\ldots,  u_n] \in \caP_\mathit{Cstr}  \ 
     \textit{s.t.}  \
 	  Str =  (ro, i)[u_1,\ldots u_j]\rho_{ro,i}\theta \}$
 \end{quote}

\noindent 
where $\&\{S_1, S_2, \ldots , S_n\}$ is a shorthand for a term $S_1 \& S_2 \& \ldots \allowbreak \& S_n$ denoting a set of strands.
$\rho_{ro,i}=\{r_1 \mapsto r_1.ro.i, \ldots, r_m\mapsto r_m.ro.i\}$ for fresh variables $r_1, \ldots, r_m$ in $[u_1,\ldots u_j, u_{j+1}, \ldots, 
u_n]$.
\end{definition}

\begin{example}
Following Examples \ref{ex:non-tlslike}  and \ref{ex:toCstrFW}, we
show a process $LP$ and a strand $Str$ that are related by the
relation $\HProc$. 
$LP$ (resp. $Str$) is the labeled process (resp. constrained strand) of the Server role after making the first explicit nondeterministic choice. \\

\begin{small}
\begin{align*}
 LP= &  (\mathit{Server, 1, 2}) \ \sigma (+(hs ; retry) \cdot -(hs ; N' ; G' ; gen(G') ; E') \cdot  \\
& ~~~   +(hs ; n(\chV{S},r1) ; G' ; gen(G') ; keyG(G',\chV{S},r_2) ;  \\
& ~~~~~~~~ Z(\chV{AReq},G',E',S,r_2,\chV{S},HM))) \\
 Str= &(\mathit{Server, 1}) \ \sigma [ \ \{?,1\} ,  -(hs ; N ; G ; gen(G) ; E) ]
\end{align*}
\end{small}

\noindent where $\sigma$ is a ground substitution to the pattern variables $N$, $G$, and $E$.
\end{example}

We then lift the duality relation between individual processes and strands to a duality relation between PA-State and FW-State.

\begin{definition}[Relation $\HPsFs$]
 Let $\Pst=\{LP_1 \& \allowbreak \ldots \& LP_n \mid \{\IK\} \}$ be a PA-State and $\Fst=\{Str_1 \& \ldots \& Str_m \allowbreak \& \{\IK'\} \}$ be a FW-State, if $(\Pst, \Fst) \in \HPsFs$, then:
\begin{itemize}
 	\item[(i)]  For each labeled process $LP_k \in \Pst$, $1\leq k \leq n$,
 	there exists a strand $Str_{k'} \in \Fst$, $1\leq k' \leq m$,
 	such that $(LP_k , Str_{k'}) \in \HProc$.
 	\item[(ii)] For each strand $Str_{k'} \in \Fst$, $1\leq k' \leq m$, there exists a labeled process $LP_k \in \Pst$, $1\leq k \leq n$, such that $(LP_k , Str_{k'}) \in \HProc$.
 \end{itemize} 
\end{definition}

The lemma below states that the relation $\HState$ induces the \emph{duality} relation $\HPsFs$.
  
 \begin{lemma}\label{lemm:equiv-PStr}
 Let $\Pst=\{LP_1 \& \ldots \& LP_n \mid \{\IK\} \}$ be a PA-State and $\Fst=\{Str_1 \& \ldots  \allowbreak \& Str_m  \& \{\IK'\} \}$ be a FW-State, if $(\Pst, \Fst) \in \HState$, i.e., exists a label sequence $\alpha$ such that
 $P_{\mathit{init}}\rightarrow_{\alpha} \Pst$, and
$F_{\mathit{init}}\rightarrow_{\alpha} \Fst$, 
then $(\Pst, \Fst) \in \HPsFs$.
 \end{lemma}
 
 \begin{proof} 
We first prove property (i). 
If $|\alpha| =0$, since both the strand set and the process configuration are empty, the statement is vacuously true. 

Now suppose that $|\alpha| > 0$. Then, without loss of generality, assume there exists a labeled process
$LP_k = ((ro,i, j) \ P_k)$ in $\Pst$, with $i, j \ge 1$. Then there is at least one label in $\alpha$ of the form $(ro, i, \_,\_,\_ )$ ($\_$ is a short hand for any content), therefore, there is a strand $St_{k'}$ in $\Fst$ of the form $(ro, i) [v_1, \ldots, v_{j'}]$. 

We then show that the above-mentioned  $LP_k$ and $Str_{k'}$ are related by $\HProc$, i.e., $(LP_k,Str_{k'}) \in \HProc$. 
Since the state $Fst$ is reachable from the initial state by the label sequence $\alpha$, and $Str_{k'}\in Fst$,  $[v_1, \ldots, v_{j'}]$ denotes exactly the sequence of messages in the unique sequence of labels $\alpha|_{(ro, i)}$. Moreover, $j'= j-1$.

Since the process state $\Pst$ is reachable from the initial state $P_{init}$ by label sequence $\alpha$, there exists a unique process  $(ro) P_{spec}$ in the specification $\ProcPA$, 
and $LP_k$ represents all possible behaviors of  $(ro) P_{spec}$ after the sequence of transitions $\alpha|_{(ro, i)}$. 
Therefore, $\toCstrFW(LP_k) =$
\begin{align*}
&\& \{ (ro, i) [u_{j}, \ldots, u_n]\rho_{ro,i}\theta  \mid \\[-.3ex]
&\exists \ \textit{ground substitution} \  \theta  \\[-.3ex]
&\exists (ro)[u_1, \ldots, u_{j-1}, u_{j}, \ldots, u_n] \in \toCstrFW((ro)P_\mathit{spec}) \\[-.3ex]
& \textit{s.t.}  \ (ro, i)[u_1, \ldots, u_{j-1}]\rho_{ro,i}\theta = (ro, i)[v_1, \ldots, v_{j-1}]\}
\end{align*}
 By the correspondence between protocol specifications defined in definition \ref{def:transf-pa-Cstr}	, $ \caP_{CstrF}  =\toCstrFW(\ProcPA)$. Also note that $(ro)P_{spec}$ is the only process in  $\ProcPA$ that has $ro$ as its role name, therefore, 
$\toCstrFW((ro)P_{spec}) = \{ (ro)[u_1, \ldots, u_n] \mid  (ro)[u_1, \ldots, u_n] \in \caP_{CstrF} \}$. 
Therefore, $\toCstrFW(LP_k) =$ 
\begin{align*}
 & \&\{ (ro, i) [u_{j}, \ldots u_n]\rho_{ro,i}\theta  \mid \\[-.3ex]
&\exists \  \textit{ground substitution} \ \theta, \\[-.3ex]
&\exists (ro)[u_1, \ldots, u_{j-1}, u_{j}, \ldots, u_n] \in \caP_{CstrF} \\[-.3ex]
&\textit{s.t.} \ (ro, i)[u_1, \ldots, u_{j-1}]\rho_{ro,i}\theta = (ro, i)[v_1, \ldots, v_{j-1}]\}. 
\end{align*}
Therefore, $(LP_k,Str_{k'}) \in \HProc$.
	  The proof for property (ii) is similar to the one for property (i).
\end{proof}

Lemma \ref{lemm:sameIK} below 
formalizes the observation that the equivalence of label sequence
implies the same intruder knowledge. 

 \begin{lemma}\label{lemm:sameIK}
 	Given a PA-State $\Pst$  and a FW-State $\Fst$ 
 such that  $(\Pst, \Fst) \in \HState$, i.e., there exists a label sequence $\alpha$ such that $P_{init}\rightarrow_{\alpha} \Pst$ and
  $F_{init}\rightarrow_{\alpha} \Fst$,
 	then the contents of intruder knowledge 
 	in $\Pst$ and in $\Fst$ are  syntactically equal.
 \end{lemma}
\begin{proof}
  In both semantics the only transition rules that add new elements
  to the intruder's knowledge are the ones whose label is of the form
  $(ro,i, j, +m,n)$. Therefore, given the two states $\Pst$ and $\Fst$ 
  as described above,  their intruder's knowledge can be computed
  from the sequence of labeled transitions $\alpha$ as 
  	\linebreak
 $\mathit{IK(\Pst) =  
  \{ 
  \inI{m} \mid  \allowbreak  (\_, \_, \_,  +m,  \_) \in \alpha \} } = \mathit{IK(\Fst) } $.
\end{proof}


 Based on the lemmas above, we can now show that the relation
 $\HState$ is a bisimulation. Since the proof of Theorem
 \ref{thm:Bisimulation} requires a somewhat lengthy case analysis,
it has been moved to Appendix \ref{bisim-proof-app}.

\begin{theorem}[Bisimulation]\label{thm:Bisimulation}
$\HState$ is a bisimulation.
\end{theorem}

\section{Constrained Backwards Strand Semantics}\label{sec:CstrBW}
In this section we extend Maude-NPA's symbolic backwards semantics with 
rules for constrained messages
of the form described in Section~\ref{sec:syntaxCstrFW}, so that
it can analyze protocols exhibiting explicit choices.
We refer to this extended backwards semantics as 
\emph{constrained backwards strand semantics}.
We then show that the \emph{constrained backwards strand semantics} is sound and complete
with respect to the 
constrained forwards strand semantics presented in Section~\ref{sec:CstrFW}, and 
the process algebra semantics presented in Section~\ref{sec:PA}. This result 
 allows us to use Maude-NPA for analyzing protocols exhibiting choice, including both implicit and explicit choices, 
 and in particular any protocol specified using the \emph{protocol process algebra}.

 
The strand space model used in the constrained backwards strand semantics
is the same as the one already used in Maude-NPA \cite{EscobarMM-fosad}, 
except for the following differences:

 \begin{itemize}
	  \item Maude-NPA explores \emph{constrained states} as defined in 
	  \cite{EscobarMMS15}, that is, states that
	  have an associated constraint store. More specifically, 
	  a \emph{constrained state}
	  is a pair
	  $\pair{St,\Psi}$ consisting of
	  a state expression  $St$ and a \emph{constraint}, i.e., a set $\Psi$ 
	  understood as a conjunction
	  $\Psi = {\bigwedge^n_{i=1}} u_i \not= v_i$ 
	  of  disequality  constraints.	  
  \item Strands are now of the form $[ u_1, \ldots, u_i \mid u_{i+1}, \ldots u_n]$,
	    where each $u_k$ can be of one of these forms: 
		  (i)  $m^+$ if it is a sent message,
		  (ii) $m^-$ if it is a received message, or
		  (iii) $\{Cstr, Num\}$ if it is a constrained message.
\end{itemize}

	State changes are described by a set $\PCstrBWRls$
    of \emph{rewrite rules}, so that the rewrite theory
	$(\PCstrSymbols,\CstrEq,\PCstrBWRls)$ 
	characterizes the behavior of protocol $\caP$ modulo the equations $\CstrEq$
	for \emph{backwards} execution. 
	The set of rules $\PCstrBWRls$ is obtained as follows.
	First, we adapt the set of rules $\PBWRls$ in 
	Section~\ref{sec:mnpa-syntax}
	to constrained states, which is an embedding of rules in
        $\PBWRls$. Their forwards version is shown below: 

\begin{small}	
		\begin{align}
		&  \pair{ \{
		SS\ 
		\&\ [L ~|~ M^-, L']\ 
		\&\ \{\inI{M},IK\}\}, \Psi}		
		\notag\\[-.5ex]
		&  \to 
		 \pair{\{
		\textit{SS}\ 
		\&\ [L, M^- ~|~ L']\ 
		\&\ \{\inI{M},IK\} 
		\}  , \Psi} \eqname{B-}
		\label{eq:negative-1-CstrBW}
		\\[1ex]
		&  \pair{ \{  
		SS\ 
		\&\ [L ~|~ M^+, L']\ 
		\&\ \{IK\}
		\} , \Psi}
		\notag\\[-.5ex]
		& \to 
		 \pair{ \{  
		SS \ 
		\&\ [L, M^+ ~|~ L']\
		\&\ \{IK\}
		\}  , \Psi} \eqname{B+}
		\label{eq:positiveNoLearn-2-CstrBW}
		\end{align}

		\begin{align}
		&  \pair{ \{  
		SS\ 
		\&\ [L ~|~ M^+, L']\ 
		\&\ \{\nI{M},IK\}
		\} , \Psi}
		\notag\\[-.5ex]
		&\to 
		 \pair{ \{  
		SS \ 
		\&\ [L, M^+ ~|~ L']\
		\&\ \{\inI{M}, IK\}
		\}  , \Psi} \eqname{B++}
		\label{eq:positiveLearn-4-CstrBW}\\[1ex]
& {	\begin{array}[t]{@{}l@{}}
	\forall \ [ l_1, u^+, l_2] \in \caP: \\
	 \pair{ \{ \{SS \, \& \, [\, l_1 |  \,  u^+, l_2\,] \,\&\, \{\nI{u},IK\}\} , \Psi} \\
	\to  \pair{ \{SS \,  \& \,  \{\inI{u},IK\}\} \} , \Psi}
	\end{array} 
	\label{eq:newstrand-positive-CstrBW}%
}\eqname{B\&}
	\end{align}%
\end{small}

\noindent
where   $L$ and $L'$ are variables denoting a list of strand messages,
$\textit{IK}$ is a variable for a set of intruder facts 
(\inI{m} or \nI{m}),
 $\textit{SS}$ is a variable denoting a set of strands,
 and $l_1$, $l_2$ denote a list of strand messages.

	Then, we define new transition rules for constrained messages.
	That is, we add  the reversed version
	of the following rules:
	
	\vspace{-2ex}
	
	\begin{small}	
	\begin{align}
		& \pair{\{ SS \, \& \, \{IK'\} \, \& \, (ro)[L \mid \{?,\mathit{Num}\}, L' ] \} , \Psi} \notag\\[-.5ex]
		& \to 
		\pair{\{ SS \, \& \, \{IK'\} \, \& \, (ro)[L, \{?,\mathit{Num}\} \mid L' ] \},\Psi}  \eqname{B?}
		\label{eq:CstrBW-nonDet}
		\end{align}

\vspace{-3ex}
		    
 \begin{align}
& \pair{\{ SS \, \& \, \{IK \} \, \& \, (ro)[L \mid \{ M =_{E_\caP} M,\mathit{Num}\}, L' ] \}, \Psi} \notag\\[-.5ex]
& \to 
\pair{\{ SS \, \& \, \{IK\} \, \& \, (ro)[L, \{M =_{E_\caP} M,\mathit{Num}\} \mid L' ] \} ,\Psi} \eqname{Bif=}
\label{eq:CstrBW-det-eq}
\end{align}

\vspace{-3ex}

\begin{align}
& \langle\{ SS \, \& \, \{IK\} \, \& \, (ro)[L \mid \{ M \neq M',\mathit{Num}\}, L' ] \} , 
         (\Psi \wedge \allowbreak M \neq M') \rangle \notag\\
& \to 
\pair{\{ SS \, \& \, \{IK\} \, \& \, (ro)[L, \{M \neq M',\mathit{Num}\} \mid L' ] \} ,\Psi}\notag\\[-.5ex]
& \mbox{if } (\Psi \wedge M \neq_{E_\caP} M') \mbox{ is satisfiable in } \caT_{\PCstrSymbols/ E_{\mathcal{P}}}(\Variables)\ 
\eqname{Bif$\mathit{\neq}$}
\label{eq:CstrBW-det-diseq}
\end{align}	
		\end{small}

Rule~\eqref{eq:CstrBW-nonDet} processes a constraint message
denoting an explicit non-deterministic choice with constant ``$?$''. The constraint store is not changed and no satisfiability
check is required.

Rules~\eqref{eq:CstrBW-det-eq} and \eqref{eq:CstrBW-det-diseq} 
deal with constrained messages 
associated to
explicit deterministic choices. 
Since the only constraints we allow in explicit deterministic choices are 
equalities and disequalities, 
rule~\eqref{eq:CstrBW-det-eq} is for the case when the constraint is an equality,  
rule \eqref{eq:CstrBW-det-diseq} is for the case when the constraint is a disequality. 
The equality constraint is solved by $E_\caP$-unification.
The constraint in a \emph{constrained state} is therefore a \emph{disequality constraint}, i.e.,  
	  $\Psi = {\bigwedge^n_{i=1}} u_i \neq_{\PEq} v_i$.  
	  The \emph{semantics} of such a constrained state, written $\semantics{\pair{St,\Psi}}$ is
	   the set of all ground substitution
	  instances  of the form:	  
	  $$\semantics{\pair{St,\Psi}} = \{ St\theta \mid \theta \in [\Variables \to \caT_{\PSymbols}] 
	  \wedge u_i\theta \neq_{E_\caP} v_i\theta, 1 \leq i \leq n  \} $$
	  The disequality constraints are then solved the same way as in \cite{EscobarMMS15}.

	\begin{definition} \label{def:CstrBW-semantics}
		Let $\mathcal{P}$ be a protocol with signature $\PCstrSymbols$ and equational theory $E_\mathcal{P}$.  We define the {\em constrained backwards rewrite theory characterizing $\caP$} to be \linebreak
		$(\PCstrSymbols, \CstrEq, \PCstrBWRls)$
		where $\CstrEq$ is same as explained in Section \ref{sec:mnpa-syntax}. $\PCstrBWRls$ is the result of reversing the rewrite rules  $\{ \eqref{eq:negative-1-CstrBW} , \eqref{eq:positiveNoLearn-2-CstrBW} ,\allowbreak \eqref{eq:positiveLearn-4-CstrBW},
		\eqref{eq:CstrBW-nonDet},
		\eqref{eq:CstrBW-det-eq}, 
		 \eqref{eq:CstrBW-det-diseq}\}  \cup \eqref{eq:newstrand-positive-CstrBW}$.
	\end{definition}

\section{Soundness and Completeness of \allowbreak Constrained Backwards Strand Semantics}
	\label{sec:proof-CstrFW-CstrBW}
The soundness and completeness proofs 
generalize the proofs in \cite{Escobar2014FW}. 
Recall that the state in the constrained states of constrained backwards strand semantics is a symbolic strand state, i.e., a state with variables. A state in the forwards strand semantics is a ground strand state, i.e., a state without variables.
The lifting relation defines the instantiation relation
 between symbolic and ground states. 
 
 We define a symbolic state and a ground state as follows.
 \begin{definition}[Symbolic Strand State]
   Given a protocol $\caP$, a \emph{symbolic} strand state $S$ of $\caP$ is a term of the form:
    \begin{align}
	 S = \{  & ::r_{1_1}, \ldots, r_{m_1}::[ u_{1_1}, \ldots u_{i_1-1} ~\mid~ u_{i_1},\ldots,  u_{n_1}] \ \& \notag\\
               & \vdots  \notag\\
               & 
                   ::r_{1_k}, \ldots, r_{m_k}:: [ u_{1_k}, \ldots, u_{i_k-1} ~\mid~ u_{i_k},\ldots,  u_{n_k}] \ \&\  SS\ \notag \\
          & \{\inI{w_1},  \ldots, \inI{w_m}, \ \nI{w'_1},  \ldots, \nI{w'_{m'}},  IK\} \} \notag
    \end{align}
    \noindent
    where for each $1\leq j\leq k$, there exists a strand
    $[ m_{1_j}, \ldots m_{i_j-1}, \allowbreak m_{i_j},\ldots,  m_{n_j}]\in \StrCstrFW$
    and a substitution $\rho_j:\Variables\to\caT_{\PSymbols}(\Variables)$ such that
    $m_{1_j}\rho_j \congr{E_\caP} u_{1_j}$,
    \ldots,
    $m_{n_j}\rho_j \congr{E_\caP} u_{n_j}$,
    $SS$ is a variable denoting a  (possibly empty) set of strands,
    and
    $IK$ is a variable denoting a (possibly empty) set of intruder's knowledge facts.
\end{definition}

\begin{definition}[Ground Strand State]
   Given a protocol $\caP$, a \emph{ground} strand state $s$ of $\caP$ is a term without variables of the form:
    \begin{align}
  s = \{  
                    & [ u_{1_1}, \ldots u_{i_1-1} ] \ \& 
                   \cdots \&\ 
                   [ u_{1_k}, \ldots, u_{i_k-1} ] \ \&\  \notag\\
            &
           \{\inI{w_1},  \ldots, \inI{w_m}\}\  \} \notag
    \end{align}
    \noindent
    where 
    for each $1\leq j\leq k$, there exists a strand
    $[ m_{1_j}, \ldots m_{i_j-1}, \allowbreak m_{i_j},\ldots,  m_{n_j}]\in \StrCstrFW$
    and a substitution $\rho_j:\Variables \to \caT_{\PSymbols}$ such that
    $m_{1_j}\rho_j \congr{E_\caP} u_{1_j}$,
    \ldots,
    $m_{i_j}\rho_j \congr{E_\caP} u_{i_j}$.
\end{definition}
  
The lifting relation in \cite{Escobar2014FW} is extended with constraints and constrained messages. 
Note that the $u_{i}$ in the definition below can be sent messages, received messages, or constrained messages. 

\begin{definition}[Lifting Relation]
    Given a protocol $\caP$, a constrained symbolic strand state $\CstrS=\cnt{S,\Psi}$ and a ground strand state $s$, 
	we say that $s$ \emph{lifts} to $\CstrS$,
	or that $\CstrS$ \emph{instantiates} to $s$ with a \emph{ground} substitution  
	$ \theta: (\var{S} - \{ \mathit{SS}, \mathit{IK} \}) \to  \caT_{\PSymbols}$,
	written $\CstrS >^\theta s$ iff 
	\begin{itemize}
		\item  for each strand $::r_{1},\ldots,r_{m}::[ u_{1}, \ldots u_{i-1} \mid u_{i},\ldots,  u_{n}]$  in $S$,
			there exists a strand 
			$[ v_{1}, \ldots v_{i-1} ]$ in $s$
 			such that $\forall 1\leq j\leq i-1$, $v_j =_{E_\caP} u_j\theta$.				
		\item  for each positive intruder fact $\inI{w}$ in $S$, 
		there exists a positive intruder fact $\inI{w'}$ in $s$ such that $w' \congr{E_\caP} w\theta$,
		and
		\item  for each negative intruder fact $\nI{w}$ in $S$, 
		there is no positive intruder fact $\inI{w'}$ in $s$ such that $w' \congr{E_\caP} w\theta$.
		\item ${E_{\caP}} \models \Psi\theta $.
	  \end{itemize}
\end{definition}
\noindent

In the following we show the soundness and completeness of transitions in constrained backwards strand semantics w.r.t. 
the constrained forwards strand
semantics by proving two lemmas stating the completeness and soundness
of one-step transition in the constrained backwards strand semantics w.r.t. 
the constrained forwards strand
semantics. The soundness and completeness result
directly follows these two lemmas. 
In the proofs we consider only transition rules added 
in both semantics to deal with explicit choices, that is, rules 
$\eqref{eq:forward-CstrMsg-det}
\cup \eqref{eq:forward-CstrMsg-nonDet} 
\cup \eqref{eq:forward-CstrMsg-first}$
in the constrained forwards strand semantics
and rules
$\{
\eqref{eq:CstrBW-nonDet},
\eqref{eq:CstrBW-det-eq},  \eqref{eq:CstrBW-det-diseq} \}$
in the constrained backwards strand semantics.
The proof 
of the soundness and completeness of
one-step transitions performed in the constrained backwards strand semantics
using rules 
$\{ \eqref{eq:negative-1-CstrBW} , \eqref{eq:positiveNoLearn-2-CstrBW}, \eqref{eq:positiveLearn-4-CstrBW}\} \cup \eqref{eq:newstrand-positive-CstrBW}$
w.r.t to one-step transitions performed in the
constrained forwards  strand semantics
using rules  $\eqref{eq:forward-positive-modIK-Cstr}
\cup \eqref{eq:forward-positive-Cstr}
\cup \eqref{eq:forward-positive2-modIK-Cstr}
\cup  \eqref{eq:forward-positive2-Cstr}
\cup  \eqref{eq:forward-negative-Cstr} 
\cup  \eqref{eq:forward-negative2-Cstr} $
is the same as in 
\cite{Escobar2014FW}, since in these transitions
no constraint is involved. 
Note that although in \cite{Escobar2014FW}, \emph{Choice Variables} were not defined explicitly, the proof extends to strands with choice variables naturally, since the lifting relation between a ground state and a symbolic state does not need to be changed to  cover choice variables. Since the strand labels are irrelevant for the result of this section, we will omit the strand labels to simplify the notation from now on.  Also, we include the fresh substitution in the substitutions and do not separate the fresh substitutions explicitly.

Extending the proofs in \cite{Escobar2014FW}, we first prove
how
the lifting of a ground state to a symbolic state induces a lifting
of a forwards rewriting step in the forwards semantics  to a backwards narrowing step
 in the backwards semantics, i.e., the completeness of one-step transition. 
 The lemma below extends the lifting lemma in \cite{Escobar2014FW} to strands with constrained messages.

\begin{lemma}[Lifting Lemma]\label{lemm: constrained-completeness}
Given a protocol $\caP$, two ground strand states  $s$ and $s'$, a constrained symbolic strand state $\CstrS'=\cnt{S', \Psi'}$
and a substitution $\theta'$ s.t.  $s \to s'$ and $\CstrS' >^{\theta'} s'$,
then there exists  a constrained symbolic strand state $\CstrS=\cnt{S, \Psi}$ and a 
substitution $\theta$ s.t.  $\CstrS >^{\theta} s$
and either  $\CstrS \narrowleft{\mu}{} \CstrS'$ or $\CstrS = \CstrS'$.
\end{lemma}
The Lifting Lemma is illustrated by Figure \ref{fig:lift}.
\begin{figure}
\centering
{
$
\xymatrix@C=13pt@R=13pt{
\CstrS\ar@{.>}[d]_{>^{\theta}} & \CstrS'\ar@{.>}[l]\ar[d]^{>^{\theta'}} \\
s\ar[r] & s'
}%
$
}
\caption{Lemma \ref{lemm: constrained-completeness}}
\label{fig:lift}

\end{figure}

 \begin{proof}
As has been explained before, we only need to consider the new rules:
\eqref{eq:forward-CstrMsg-det}, \eqref{eq:forward-CstrMsg-nonDet},
\eqref{eq:forward-CstrMsg-first}. The proof in \cite{Escobar2014FW} is
structured by
 cases, some of which having  specific requirements on intruder knowledge, 
or involve changes made to the intruder knowledge. Since all the new
rules we are 
considering do not have specific requirements on the intruder
knowledge, and do 
not change the intruder knowledge either, the cases that we need
to consider are 
the following (cases $e$ and $f$ in the proof  in
\cite{Escobar2014FW}), which 
involve the appearance or non-appearance of certain strand(s):
\begin{itemize}
		\item[e:] 
There is a strand $[u_1,\ldots,u_{j-1},u_j,\ldots,u_n]$ in $\StrCstrFW$,
$n \geq 1$, $1 \leq j \leq n$,
and a substitution $\rho$ such that
$[u_1,\ldots,u_{j-1},u_j]\rho$ is a strand in $s'$
and
$[u_1,\ldots,u_{j-1}, u_j \mid u_{j+1},\ldots,u_n]\rho$ is a strand in $S'\theta'$.

\item[f:] 
There is a strand $[u_1,\ldots,u_{j-1},u_j,\ldots,u_n]$ in $\StrCstrFW$,
$n \geq 1$, $1 \leq j \leq n$,
and a substitution $\rho$ such that
$[u_1,\ldots,u_{j-1},u_j]\rho$ is a strand in $s'$
but
$[u_1,\ldots,u_{j-1},  u_j \mid u_{j+1} , \ldots,u_n]\rho$ is not a strand in $S'\theta'$.
		\end{itemize} 		
		
Now we consider for the forward rewrite rule application in the step $s\rightarrow s'$.  

\begin{itemize}
\item Given ground states $s$ and $s'$ s.t. $s\rightarrow s'$ using a rule in set \eqref{eq:forward-CstrMsg-det}, then there exists a ground substitution $\tau$, variables SS' and IK', and strand $ [ u_1,\ldots,u_{j-1}, \linebreak \{T, Num\}, u_{j+1},\ldots,u_n]$ in $\StrCstrFW$, such that $s =\{SS'\tau  \& \{IK'\tau\} \& (ro) [u_1\tau,\allowbreak \ldots,u_{j-1}\tau]\} $, and $s'=\{SS'\tau  \&  \{IK'\tau\} \& [u_1\tau,\ldots,u_{j-1}\tau, \{T\tau, Num\}]\} $ and $T\tau =_{E_{\caP}} true$. Since there exists a substitution $\theta '$ s. t.  $\CstrS' >^{\theta'} s'$, we consider the following two cases:

\begin{itemize}
\item{Case e)} The strand appears in $S'\theta'$. More specifically, 
$[ u_1\sigma,\ldots, u_{j-1}\sigma, \linebreak \{T\sigma, Num\} \mid u_{j+1}\sigma,\ldots,u_n\sigma]$ is a strand in $S'$ 
s.t. $\sigma\theta' =_{E_\caP} \tau$.  
If the constraint $T$ is an equality constraint, 
since $T\tau =_{E_\caP} T\sigma\theta'  =_{E_{\caP}} true$, 
and by the lifting relation, $E_\caP \models \Psi'\theta'$, 
rule \eqref{eq:CstrBW-det-eq} can be applied for the backwards narrowing $\CstrS' \narrowleft{\mu}{} \CstrS$, 
and $\CstrS >^{\theta} s$ such that $\mu\theta =_{E_\caP}\theta'$.
If the constraint $T$ is a disequality constraint, 
since $T\tau =_{E_\caP} T\sigma\theta'  =_{E_{\caP}} true$, 
and by the lifting relation, $E_\caP \models \Psi'\theta'$, 
we have $E_\caP \models T\sigma\theta' \wedge \Psi'\theta'$.
 Therefore, rule \eqref{eq:CstrBW-det-diseq} can be applied for the backwards narrowing, and $\CstrS >^{\theta} s$. 

\item {Case f)} 	The strand does not appear in $S'\theta'$. Then $\theta'$ makes $S'$ as a valid symbolic strand state of $s$, i.e., $S = S'$ and $\CstrS' >^{\theta'} s$.
	\end{itemize}

\item Given ground strand states $s$ and $s'$ s.t. $s\rightarrow s'$ using a rule in set \eqref{eq:forward-CstrMsg-nonDet}, then we consider the following two applicable cases:
\begin{itemize}
		\item {Case e)} The strand appears in $S'\theta'$
		and thus we can perform a backwards narrowing step from $\CstrS'$
		with rule~\eqref{eq:CstrBW-nonDet}, i.e., $\CstrS'  \narrow{}{} \CstrS$,
		 and $\CstrS  >^{\theta'} s$.

		\item {Case f)} The strand does not appear in $S'\theta'$.
		Then $\theta'$ makes  $\CstrS'$ as a valid constraint symbolic state of $s$,
		i.e., $\CstrS = \CstrS'$ and $\CstrS >^{\theta'} s$.
		\end{itemize}
		
\item  Given states $s$ and $s'$ s.t. $s\rightarrow s'$ using a rule in set \eqref{eq:forward-CstrMsg-first}, the proof is similar with using a rule in the set \eqref{eq:forward-CstrMsg-nonDet}.
\end{itemize}
\end{proof}
  
 Theorem \ref{thm:completeness} below then follows straightforwardly. 
 
\begin{theorem}[Completeness]\label{thm:completeness}
 	Given a protocol $\caP$,  two ground strand states $s, s_0$,
	a constrained symbolic strand state $\CstrS$ and a substitution $\theta$ s.t.
	(i) $s_0$ is an initial state, 
	(ii) $s_0 \to^n s$, and
	(iii) $\CstrS >^\theta s$. 
	Then there exists a constrained symbolic initial strand state $\CstrS_0$, two substitutions  $\mu$  and $\theta'$, 
	and $k \leq n$, s.t. 
	$\CstrS_0  \narrowleft{k}{\mu} \CstrS$, and
	 $\CstrS_0 >^{\theta'} s_0$.
\end{theorem}

The Soundness Theorem from \cite{Escobar2014FW} can also be extended to constrained backwards and forwards strand semantics. 
We first show that Lemma 2 in \cite{Escobar2014FW}, which states the soundness of one-step transition, still holds after extending to constrained states. The Soundness Theorem then follows straightforwardly.

\begin{lemma}\label{lemm: constrained-soundness} 
Given a protocol $\caP$, two constrained symbolic states $\CstrS=\cnt{S, \Psi}$ and $\CstrS'=\cnt{S', \Psi' }$,
a ground strand state $s$ 
and a ground substitution $\theta$, if
$\CstrS  \narrowleftTop{\mu} \CstrS'$ and   $\CstrS >^\theta s$, 
then there exists a ground strand state $s'$ and a ground substitution $\theta'$ such that
$s  \to s'$, and $\CstrS' >^{\theta'} s'$.
\end{lemma}

Lemma \ref{lemm: constrained-soundness} is illustrated by the Figure \ref{fig:soundness}.
\begin{figure}
\centering
{
$
\xymatrix@C=13pt@R=13pt{
\CstrS\ar[d]_{>^{\theta}} & \CstrS'\ar@{~>}[l]\ar@{.>}[d]^{>^{\theta'}} \\
s\ar@{.>}[r] & s'
}%
$
}
\caption{Lemma \ref{lemm: constrained-soundness}}
\label{fig:soundness}
\end{figure}

\begin{proof}
We only need to consider the new rules: rule \eqref{eq:CstrBW-det-eq}, \eqref{eq:CstrBW-det-diseq} and \eqref{eq:CstrBW-nonDet}. \\
1) If $\CstrS  \narrowleftTop{\mu} {\CstrS'}$ using rule \eqref{eq:CstrBW-nonDet}, then there are associated rules in the sets \eqref{eq:forward-CstrMsg-nonDet} and \eqref{eq:forward-CstrMsg-first}.\\
2) If $\CstrS  \narrowleftTop{\mu} {\CstrS'}$ using rule \eqref{eq:CstrBW-det-eq}, 
there is a strand  
$[ u_1\sigma,\ldots,u_{j-1}\sigma \mid \{(u=v)\sigma, Num\}, u_{j+1}\sigma,\ldots, u_n\sigma]$ in $S$, 
$[ u_1\sigma',\ldots,u_{j-1}\sigma', \{(u=v)\sigma', Num\} \mid u_{j+1}\sigma',\ldots, \allowbreak u_n\sigma']$ in $S'$
s.t. $\sigma=_{E_\caP} \sigma'\mu$, $\Psi =_{E_\caP} \Psi' \mu$ and $u\sigma=_{E_\caP}v\sigma$, 
where $ [ u_1,\ldots,u_{j-1}, \{u=v, Num\},u_{j+1},\ldots,u_n]$ is a strand in $\StrCstrFW$.
  Since $\CstrS >^\theta s$, there is a ground strand $[ u_1\sigma\theta,\ldots, \allowbreak u_{j-1}\sigma\theta]$ in s, 
  and $E_\caP \models \Psi\theta $. 
  Therefore, $E_\caP \models \Psi' \mu\theta $ and $u\sigma\theta=_{E_\caP}v\sigma\theta$.
  By rule 
\eqref{eq:forward-CstrMsg-det}, $s\rightarrow s'$, and  $\CstrS' >^{\mu\theta} s'$.
  
  If $\CstrS  \narrowleftTop{\mu} {\CstrS'}$ using rule \eqref{eq:CstrBW-det-diseq}, 
there is a strand  
$[ u_1\sigma,\ldots,u_{j-1}\sigma \mid \{(u\neq v) \sigma, Num\}, u_{j+1}\sigma,\ldots, u_n\sigma]$ in $S$ , 
$[ u_1\sigma',\ldots,u_{j-1}\sigma', \{(u\neq v)\sigma', Num\} \mid  u_{j+1}\sigma', \ldots, \allowbreak u_n\sigma']$ in $S'$
s.t. $\sigma=_{E_\caP} \sigma'\mu$ and $\Psi =_{E_\caP} \Psi'\mu \wedge (u\neq v) \sigma'\mu$, 
where $ [ u_1,\ldots,u_{j-1}, \{u\neq v, Num\},u_{j+1},\ldots,u_n]$ is a strand in $\StrCstrFW$.
  Since $\CstrS >^\theta s$, there is a ground strand $[ u_1\sigma\theta,\ldots, \allowbreak u_{j-1}\sigma\theta]$ in s, 
  and $E_\caP \models \Psi\theta$. Therefore, $E_\caP \models \Psi'\mu\theta \wedge (u\neq v) \sigma'\mu\theta$.
 By rule 
\eqref{eq:forward-CstrMsg-det}, $s\rightarrow s'$, and  $\CstrS' >^{\mu\theta} s'$.
\end{proof}

The Soundness Theorem below shows that the backwards symbolic reachability analysis
is \emph{sound} with respect to the forwards rewriting-based strand semantics.

\begin{theorem}[Soundness]\label{thm:soundness}
Given a protocol $\caP$, two constrained symbolic strand states $\CstrS_0, \CstrS'$,
	an initial ground strand state $s_0$ and a substitution $\theta$ s.t.
	(i) $\CstrS_0$ is a symbolic initial state, and
	(ii) $\CstrS_0  \stackrel{*}{\leftsquigarrow} \CstrS'$    , and
	(iii)  $\CstrS_0 >^{\theta} s_0$. 
	Then there exists a ground strand state $s'$ and a substitution $\theta'$, s.t.
	(i)  $s_0  \to^* s'$, and
	(ii)  $\CstrS' >^{\theta'} s'$.
\end{theorem}

The soundness and completeness results in Theorems~\ref{thm:soundness} and~\ref{thm:completeness} together with the bisimulation proved in Theorem~\ref{thm:Bisimulation} show that the backwards symbolic reachability analysis
is \emph{sound}  and \emph{complete} with respect to the process algebra semantics.

\begin{theorem}[Soundness]\label{thm: Soundness-PA}
Given a protocol $\caP$,  two constrained symbolic strand states $\CstrS_0, \CstrS$, the initial FW-State $F_{init}$, a substitution $\theta$, and the initial PA-State $P_{init}$ s.t.
	(i)   $\CstrS_0$ is a symbolic initial strand state, and
	(ii)  $\CstrS_0  \stackrel{*}{\leftsquigarrow}_{\mu} \CstrS$, and
	(iii) $\CstrS_0 >^{\theta} F_{init}$.  
Then there exists a FW-State $\Fst$ 
such that $\CstrS >^{\theta'} \Fst$,  and therefore, there is a PA-State $\Pst$ such that $\Pst\  \HState\   \Fst$.
\end{theorem}



\begin{theorem}[Completeness]\label{thm: Completeness-PA}
%
Given a protocol $\caP$, a PA-State $\Pst$, a FW-State $\Fst$, a constrained symbolic strand state $\CstrS$ s.t.  
(i) $\Pst\ \HState\ \Fst$, (ii)  $\CstrS >^{\theta'} \Fst$. Then there is a backwards symbolic execution $\CstrS_0  \stackrel{*}{\leftsquigarrow}_{\mu} \CstrS$ s.t. $\CstrS_0$ is a symbolic initial strand state as defined in Section \ref{sec:mnpa-syntax}, and $\CstrS_0 >^{\theta} F_{init}$. 
\end{theorem}



\section{Protocol Experiments} \label{sec:Exp}

In this section we describe some 
experiments\footnote{Available at \url{http://personales.upv.es/sanesro/Maude-NPA-choice/choice.html}} that we have performed on protocols with choice. 
We have fully integrated the process algebra syntax,
and its transformation into strands, and have developed new methods
to specify attack states using the process notation  in the recent release of Maude-NPA
3.0 (see 
\cite{MNPAmanual3.0}).

\subsection{Integration of the Protocol Process Algebra in Maude-NPA}
\label{Integration}

We have fully implemented the process algebra notation in Maude-NPA.
Strands represent each role behavior as a linear sequence of message outputs and inputs 
but processes represent each role behavior as a possibly non-linear sequence of message outputs and inputs.
The honest principal specification is specified in the process algebra syntax. 
In order for Maude-NPA to accept process specifications, we have replaced the section 
\texttt{STRANDS-PROTOCOL} from the protocol template 
by a new section 
\texttt{PROCESSES-PROTOCOL}; see \cite{MNPAmanual3.0} for details.
The intruder capabilities as well as the states generated by the tool still use the strand syntax. 

Attack patterns may be specified using the process algebra syntax,
under the label \texttt{ATTACK-PROCESS},
 or strand syntax,
under the label \texttt{ATTACK-STATE}.
We describe how they are specified in the process algebra syntax below.
An attack pattern describes a state consisting  
of zero or more processes that must have executed, and zero or more terms
in the intruder knowledge.  It may also contain \emph{never patterns},
that is, descriptions of 
processes that must \emph{not} be
executed at the time the state is reached.  Never patterns can be used
to reason about authentication properties, e.g., can
Alice execute an instance of the protocol, apparently with Bob, without Bob executing an instance of the protocol with Alice.

Note that processes in an attack pattern cannot contain explicit nondeterminism (?) or explicit deterministic choice (if), 
since one and only one behavior is provided in an attack pattern. 
This is achieve by requiring that any constraint $c$ appearing in an attack pattern must be \emph{strongly irreducible}, that is, it
must not only be irreducible, but for any irreducible substitution $\sigma$ to the variables of $c$, $\sigma c$ must be irreducible as well.

That is, imagine a process i the form
$$-(m1)\ .\ +(m2)\ .\ if\ exp1\ =\ exp2\ then\ +(m3)\ else\ +(m4)$$
\noindent where each of the expressions $exp1$ and  $exp2$ can evaluate to $yes$ or $no$ depending on the substitutions made to them.

Then in the attack pattern one must specify one and only one of the following possibilities
\begin{eqnarray}
-(m1)\ .\ +(m2)\ .\ yes=yes\ .\ +(m3) \notag \\
-(m1)\ .\ +(m2)\ .\ yes \neq no\ .\ +(m4) \notag \\
-(m1)\ .\ +(m2)\ .\ no = no\ .\ +(m3) \notag \\
-(m1)\ .\ +(m2)\ .\ no \neq yes\ .\ +(m4) \notag 
\end{eqnarray}

Finally, never patterns must satisfy a stronger condition: the entire never pattern must be strongly irreducible.  This condition
is inherited from the original Maude-NPA.

\subsection{Choice of Encryption Type}
This protocol allows either public key encryption or shared key
encryption to be used by Alice to communicate with Bob. Alice
initiates the conversation by sending out a message containing the
chosen encryption mode, then Bob replies by sending an encrypted
message containing his session key. The encryption mode is chosen
nondeterministically by Alice. Therefore, it exhibits an
\emph{explicit nondeterministic choice}. Below we show the protocol
description: the first one reflects the case in which public key
encryption (denoted by $PubKey$) is chosen.

 			\begin{enumerate}
 				\item $A \rightarrow B: A ; B ;\mathit{PubKey} $ 
 				\item $B \rightarrow A:  pk(A,  B; SK)$
 				\item $A\rightarrow B:   pk(B,  A ;  SK ; N_A\}$
 				\item $B \rightarrow A : pk(A,  B ; N_A )$
 			\end{enumerate}


 	\noindent
 			The second one reflects the case
 	in which a shared key encryption (denoted by $SharedKey$) is chosen.

 			\begin{enumerate}
 				\item $A \rightarrow B: A ; B ;\mathit{SharedKey} $
 				\item $B \rightarrow A:  shk(key(A,B),  B ; SK)$
 				\item $A\rightarrow B:  shk(key(A,B), A ; SK ; N_A)$
 				\item $B \rightarrow A : shk(key(A,B),  B ; N_A)$
 			\end{enumerate}		


 	\noindent
 	Note that $A$ and $B$ are names of principals, $SK$ denotes the session key
 	generated by $B$, and $N_A$ denotes
 	a nonce generated by $A$.

There are different ways of encoding this protocol as two process expressions.
We have chosen to treat the encryption mode as a choice variable which can be either public key encryption or shared key encryption, 
and then the receiver will perform
 an explicit deterministic choice depending on the value of this choice variable. 
The process specification is as follows:

 \begin{align*}
 	 (\mathit{Init}) \  ( & (+(\chV{A} \, ; \, \chV{B} \, ; \,\mathit{PubKey}) \cdot 
						        -(pk(\chV{A}, \chV{B}  \, ; \, \mathit{SK}))  \\[-.6ex]
 	& ? \\[-.5ex]
 	   &                   (+(\chV{A} \,;\, \chV{B} \,;\, \mathit{SharedKey}) \cdot 
 	                       -(e(\mathit{key}(\chV{A},\chV{B}), \chV{B} \, ;\, \mathit{SK}))   \\
  	 (\mathit{Resp}) \   & -(A \, ; \,  B  \, ; \, \mathit{TEnc}) \  \cdot \\[-.6ex]
					  	   & \mathit{if} \  TEnc = \mathit{PubKey} \\[-.6ex]
					  	   & \hspace{3ex}  \mathit{then }   
					  	     \ (+(pk(A, B \, ; \, \mathit{skey(A,B,r')})) \    \\[-.6ex]
					  	&    \hspace{3ex} \mathit{else } \ 
  	                         (+(e(\mathit{key(A,B)}, B \, ; \, \mathit{skey(A,B,r')}))) \   
 \end{align*}

We analyzed whether the intruder can learn the session key generated by Bob, when either the public key encryption or shared key encryption is chosen, assuming both principals are honest. 

\begin{verbatim}
--- initiator accepts session key for shared key encryption and 
--- intruder learns it
  eq ATTACK-PROCESS(2)
   = -(a ; b ; mode) .
     (mode neq pubkey) .
     +(she(key(a, b), skey(b,r))) .
     -(she(key(a, b), skey(b,r) ; N)) .
     +(she(key(a, b), N))
     || skey(b,r) inI
     || nil
 [nonexec] . 

--- initiator accepts session key for public key encryption and  
--- intruder learns it
  eq ATTACK-PROCESS(3)
   = -(a ; b ; mode) .
     (mode eq pubkey) .
     +(pk(a, b ; skey(b, r))) .
     -(pk(b, a ; skey(b,r) ; N)) .
     +(pk(a, b ; N)) 
     || skey(b,r) inI
     || nil
 [nonexec] . 
\end{verbatim}

For this property, Maude-NPA terminated without any attack being found for any of the two attack states.

\subsection{Rock-Paper-Scissors}

To evaluate our approach on protocols with explicit deterministic
choices, we have used a simple protocol which simulates the famous
Rock-Paper-Scissors game, in which Alice and Bob are the two players
of the game. In this game, Alice and Bob commit to each other their
hand shapes, which are later on revealed to each other after both
players committed their hand shapes. The result of the game is then
agreed upon between the two players according to the rule: rock beats
scissors, scissors beats paper and paper beats rock.  They finish by
verifying with each other that they both reached the same conclusion.
Thus, at the end of the protocol each party should know the outcome of
the game and whether or not the other party agrees to the outcome.
This protocol exhibits \emph{explicit deterministic choice}, because
the result of the game depends on the evaluation of the committed hand
shapes according to the game's rule. Note that this protocol also
exhibits \emph{implicit nondeterministic choice}, since the hand shape
of the players are chosen by the players during the game.

The protocol proceeds as follows. First, both initiator and responder choose their hand shapes and send them to each other using a secure commitment scheme.
Next, they both send each other the nonces that are  necessary to open the commitments. Each of them 
then compares the two hand shapes and decides if the initiator
wins, the responder wins, or there is a tie.  The initiator then sends the responder the outcome.  When the responder receives the initiator's verdict, it compares
it against its own.  It responds with ``finished'' if it agrees with the initiator and ``cheater'' if it doesn't.  All messages are signed and encrypted, and the
initiator's and responder's nonces are included in the messages concerning the outcome of the game.  The actual messages sent and choices made are
described in more detail below.  


\clearpage
\begin{enumerate}
 				\item $A \rightarrow B:  pk(B,  sign(A, commit(N_A, X_A) ) ) $
 				\item $B \rightarrow A:  pk(A, sign(B, commit(N_B, X_B ) )  )$
 				\item $A \rightarrow B:  pk(B, sign(A, N_A )) $
 				\item $B \rightarrow A:  pk(A, sign(B,  N_B)) $
 				\item $if \ (X_A \ beats  \ X_B) \  $\\
 				               $ then \ R = Win $ \\ 
 				  $ else \ if  \ (X_B \ beats \ X_A)$ \\
 				  \hspace*{5mm}            $ then \ R = Lose$ \\ 
 				  \hspace*{5mm}            $   else\ if  \ (X_B \ = \ X_A)  \\ 
				  \hspace*{13mm} then \ R = \ Tie$  \\
 				  \hspace*{13mm} $ else\ \nil$ 
				\item $A \rightarrow B : pk(B,  sign(A, N_A ; N_B ; R )  ) $ 
				 \item $if \ (R =  Win \& X_A \ beats \ X_ B )  $ \\
				 	\hspace*{16mm} $ \ or \ (R \ = \ Lose \  \& \ X_B  \ beats \  X_A )$\\
					 \hspace*{16mm} $\  or \ (R \ =  \ Tie \ \&  \  X_A \ = \ X_B) $\\
					  \hspace*{6mm}      $ then \ B \rightarrow A:  pk(A, sign(B, N_A ; N_B ; \ finished \  )  )$ \\
				      \hspace*{6mm}      $ else\ B \rightarrow A:  pk(A, sign(B, N_A ; N_B ; \ cheater  \  )  )$
\end{enumerate}		


One interesting feature of the Rock-Scissors-Paper protocol, is that, in order
to  verify that the commitment has been opened successfully, i.e., that the nonce received is the
nonce used to create the commitment, one must verify that the result of opening it is well-typed, i.e., that it is
equal to ``rock'', ``scissors'', or ``paper''.  This can be done via the evaluation of predicates.  First, we create a sort Item and declare the constants ``rock'', ``scissors'', and ``paper''
to be of sort Item.  Then we create a variable $X:Item$ of sort Item.  We then define a predicate $item?$ such that $item? X:Item $ evaluates to true.  Since only terms of sort
Item can be unified with $X:Item$, this predicate can be used to check whether or not a term is of sort Item.
The process specification for the initiator and the responder is as follows.

 \begin{align*}
(Initiator) \ 
&          +(pk(B, sig(A, com(n(A,r), XA))))\ .\\
&             -(pk(A, sig(B, ComXB)))\ .\\
&             +(pk(B, sig(A, n(A , r))))\ .\\
&             -(pk(A, sig(B, NB)))\ .\\
&              (if\ ((item?\ open(NB, ComXB))\ eq\  ok)\\ 
&  then\  
                if\ ((XA\ beats\ open(NB, ComXB))\ eq\ ok)\\
&  ~~~~~~~                then\ +(pk(B, sig(A, n(A, r)\  ;\ win))) \\
&  ~~~~~~~                else\ if\ ((open(NB, ComXB)\ beats\ XA)\ eq\ ok)\\
&  ~~~~~~~~~~~~                        then\ +(pk(B, sig(A, n(A, r)\ ;\ lose )))\\
&  ~~~~~~~~~~~~                        else\ +(pk(B, sig(A, n(A, r)\ ;\ tie)))\\
&              else\ nilP )\ .\\
&	     -(pk(A, sig(B, n(A,r) ; NB))\ ;\ S{:}Status)
 \end{align*}
 
 \begin{align*}
(Responder) \
&             -(pk(B, sig(A, ComXA)))\ .\\
&             +(pk(A, sig(B, com(n(B,r), XB))))\ .\\
&             -(pk(B, sig(A, NA)))\ .\\
&             +(pk(A, sig(B, n(B, r))))\ .\\
&             -(pk(B, sig(A, NA ; R)))\ .\\
&              (if\ ((item?\ open(NA, ComXA))\ eq\ ok)\ then\\
&	         if\ (R\ eq\ win)\\
&		 then\ if\ ((open(NA, ComXA)\ beats\ XB)\ eq\ ok) \\
& ~~~~~~                     then\ +(pk(A, sig(B, NA ; n(B,r)))\ ;\ finished)\\
& ~~~~~~		      else\ +(pk(A, sig(B, NA ; n(B,r)))\ ;\ cheater)\\
&		 else\ if\ (R\ eq\ lose)\\
& ~~~~~~		      then\ if\ ((XB\ beats\ open(NA, ComXA))\ eq\ ok)\\
& ~~~~~~~~~~~~		           then\ +(pk(A, sig(B, NA ; n(B,r)))\ ;\ finished)\\
& ~~~~~~~~~~~~		           else\ +(pk(A, sig(B, NA ; n(B,r)))\ ;\ cheater)\\
&~~~~~~		      else\ if\ (R\ eq\ tie)\\
&~~~~~~~~~~~~		           then\ if\ (XB\ eq\ open(NA, ComXA))\\
&~~~~~~~~~~~~~~~~~~			        then\ +(pk(A, sig(B, NA ; n(B,r)))\ ;\ finished)\\
&~~~~~~~~~~~~~~~~~~				else\ +(pk(A, sig(B, NA ; n(B,r)))\ ;\ cheater)\\
&~~~~~~~~~~~~		           else\ nilP\\
&               else\ nilP )
 \end{align*}

 	\noindent
We first tried to see whether the protocol can simulate the game successfully, so we asked for different  scenarios in which the player Alice or Bob can win in a round of the game. Maude-NPA was able to generate the expected scenarios, and it did not generate any others. We then gave Maude-NPA a secrecy attack state,
in which the intruder, playing the role of initiator  against an honest responder, attempts to  guess its nonce before the responder receives  its commitment.

\begin{verbatim}
eq ATTACK-PROCESS(1) =
   -(pk(b,sig(i, ComXA:ComMsg))) .
   +(pk(i,sig(b, com(n(b, r:Fresh), XB:Item)))) .
   -(pk(b,sig(i, NA:Nonce)))
   || n(b, r:Fresh) inI
   || nil
[nonexec] .
\end{verbatim}

Finally we specified an authentication attack state in which we asked if a responder could complete a session with an honest initiator with the conclusion
that the initiator had carried out its rule faithfully, without that actually having happened.

\begin{verbatim}
eq ATTACK-PROCESS(2) =
  -(pk(b, sig(a, ComXA))) .
  +(pk(a, sig(b, com(n(b,r), XB)))) .
  -(pk(b, sig(a, NA))) .
  +(pk(a, sig(b, n(b, r)))) .
  -(pk(b, sig(a, NA ; win))) .
  ((item? open(NA, ComXA)) eq ok) .
  (win eq win) .
  ((open(NA, ComXA) beats XB) eq ok) .
  +(pk(a, sig(b, NA ; n(b,r))) ; finished)
  || empty
  || never(
     +(pk(b, sig(a, ComXA))) .
     -(pk(a, sig(b, com(n(b,r), XB)))) .
     +(pk(b, sig(a, NA))) .
     -(pk(a, sig(b, n(b,r)))) .
     (ok eq ok) .
     (ok eq ok) .
     +(pk(b, sig(a, NA ; win))) .
     -(pk(a, sig(b, NA ; n(b,r))) ; finished)
     || empty
  )
[nonexec] .
\end{verbatim}

For both of these attack states Maude-NPA finished its search without finding any attacks.


\subsection{TLS}
 In Section \ref{sec:ex} we introduced a simplified version of the handshake protocol in TLS 1.3 \cite{tls1.3.12}.
 Even this simplified version produced a very large search space, because of the long list of messages and the concurrent interactions of a big amount of  choices.
We are however able to 
 check the correctness of our specification by producing legal executions
 in Maude-NPA. 
Unlike TLS 1.3, we intentionally introduced a ``downgrade attack"  in our version in which the attacker can trick the principals into using a weaker crypto system. 
However, we have not yet been able to produce this attack because of the very deep and wide analysis tree (i.e., long reachability sequences with many  branches) that is produced.  
We are currently investigating more efficient ways of managing list processing. 

\section{Related Work}\label{sec:related-work}

As we mentioned in the introduction, there is a considerable amount of
work on adding choice to the strand space model that involves
embedding it into other formal systems, including event-based models
for concurrency \cite{DBLP:conf/fsttcs/CrazzolaraW02}, Petri nets
\cite{DBLP:journals/entcs/Froschle09}, or multi-set rewriting
\cite{DBLP:conf/csfw/CervesatoDMLS00}.  Crazzolara and Winskel model
nondeterministic choice as a form of composition, where a conflict
relation is defined between possible child strands so that the parent
can compose with only one potential child.  In
\cite{DBLP:journals/entcs/Froschle09} Fr{\"{o}}schle uses a Petri net
model to add branching to strand space \emph{bundles}, which represent
the concurrent execution of strand space roles.  Note that we have
taken the opposite approach of representing bundles as traces of
non-branching strands, where a different trace is generated for each
choice taken.  Although this results in more bundles during forward
execution, it makes little difference in backwards execution, and is
more straightforward to implement in an already existing analysis
tool. 
 
 We also note that deterministic choice has been included in the applied pi calculus for cryptographic protocols \cite{DBLP:conf/popl/AbadiF01}, another widely used formal model, based on Milner's pi calculus \cite{DBLP:books/daglib/0098267}.
 The applied pi calculus includes the rule $\mathit{if} \ M = N \ \mathit{then} \ P \  \mathit{else} \ Q$, where $P$ and $Q$ are terms.  This is similar to our  syntax for deterministic choice.  However
 our long-term plan is to add other types or predicates as well (e.g., 
$M \ \mathit{subsumes} \ N$) ; indeed our approach extends to any type of predicate that can be evaluated
 on a ground state.
 Although the applied pi calculus in its original form does not include nondeterministic choice, both nondeterministic and probabilistic  choice have
 been added in subsequent work \cite{DBLP:conf/aplas/Goubault-LarrecqPT07}.
 
In addition, Olarte and Valencia show in
\cite{DBLP:conf/ppdp/OlarteV08} how a cryptographic protocol modeling
language can be expressed in their universal timed concurrent
constraint programming (utcc) model, a framework that extends the
timed concurrent constraint programming model to express mobility.
The language does not support choice, but utcc does.  It
 seems that it would not be difficult to extend the language 
to incorporate the utcc choice mechanisms.

 The Tamarin protocol analysis tool \cite{DBLP:conf/cav/MeierSCB13} includes deterministic branching, which was used extensively in the analysis  of TLS 1.3 \cite{tls-tamarin/2016}.  In particular,
 it includes an optimization for roles of the form $P . ( \mathit{if} \ T \ \mathit{then} \ {Q} \ \mathit{else} \ R) . S$; when backwards search is used, it is sometimes possible
 to capture such an execution in terms of just one strand until the conditional is encountered, thus reducing the  state space.  Our approach produces two
 strands, but since the process algebra semantics makes it easy to tell whether or not $R$ behaves ``essentially'' the same no matter if $P$ or $Q$ is chosen,
 we believe that  we have a pathway for including such a feature if desired.

\section{Conclusions} \label{sec:Conclusion}

We have provided an extension to the strand space model that allows
for both deterministic and nondeterministic \allowbreak choice, together with an
operational semantics for choice in strand spaces that not only
provides a formal foundation for choice, but allows us to implement it
directly in the Maude-NPA cryptographic protocol analysis tool.  In
particular, we have applied Maude-NPA to several protocols that rely
on choice in order to validate our approach.

This work not only provides a choice extension to strand spaces, but
extends them in other ways as well.  First of all, it provides a process
algebra for strand spaces.  This potentially allows us to relate the
strand space model to other formal systems (e.g., the applied pi
calculus \cite{DBLP:conf/podc/Abadi01}) giving a better understanding of how it compares
with other formal models.  In addition, the process algebra semantics
provides a new specification language for
Maude-NPA that we  believe is more natural for users than
the current strand-space
language.  

Another contribution of this work is that it provides a means for evaluating both equality and disequality predicates in the strand space model and in Maude-NPA.  This
allows us to   implement features such as type checking in Maude-NPA, via predicates such as $\mathit{foocheck(X)}$, where \allowbreak
$\mathit{foocheck(0:Foo)} = tt$, that is, $\mathit{foocheck(X)}$
succeeds only if $X$ is of sort $Foo$.  This proved to be very helpful, for example, in our specification of the Rock-Scissors-Paper protocol as we described earlier.
We believe the expressiveness of Maude-NPA can be further increased at little cost by extending the types of predicates that can be
evaluated, e.g., by including predicates for subsumption and their
negations.
This is another subject for further investigation.


\begin{thebibliography}{18}
\providecommand{\natexlab}[1]{#1}
\providecommand{\url}[1]{{#1}}
\providecommand{\urlprefix}{URL }
\expandafter\ifx\csname urlstyle\endcsname\relax
  \providecommand{\doi}[1]{DOI~\discretionary{}{}{}#1}\else
  \providecommand{\doi}{DOI~\discretionary{}{}{}\begingroup
  \urlstyle{rm}\Url}\fi
\providecommand{\eprint}[2][]{\url{#2}}

\bibitem[{Abadi(2001)}]{DBLP:conf/podc/Abadi01}
Abadi M (2001) Leslie lamport's properties and actions. In: Proceedings of the
  Twentieth Annual {ACM} Symposium on Principles of Distributed Computing,
  {PODC} 2001, p~15

\bibitem[{Abadi and Fournet(2001)}]{DBLP:conf/popl/AbadiF01}
Abadi M, Fournet C (2001) Mobile values, new names, and secure communication.
  In: Conference Record of {POPL} 2001: The 28th {ACM} {SIGPLAN-SIGACT}
  Symposium on Principles of Programming Languages, pp 104--115

\bibitem[{Cervesato et~al(2000)}]{DBLP:conf/csfw/CervesatoDMLS00}
Cervesato I, Durgin NA, Mitchell JC, Lincoln P, Scedrov A (2000) Relating
  strands and multiset rewriting for security protocol analysis. In:
  Proceedings of the 13th {IEEE} Computer Security Foundations Workshop, {CSFW}
  '00, pp 35--51

\bibitem[{Crazzolara and Winskel(2002)}]{DBLP:conf/fsttcs/CrazzolaraW02}
Crazzolara F, Winskel G (2002) Composing strand spaces. In: {FST} {TCS} 2002:
  Foundations of Software Technology and Theoretical Computer Science, pp
  97--108

\bibitem[{Cremers et~al(2016)}]{tls-tamarin/2016}
Cremers C, Horvat M, Scott S, van~der Merwe T (2016) Automated analysis and
  verification of {TLS} 1.3: 0-rtt, resumption and delayed authentication. In:
  {IEEE} Symposium on Security and Privacy, {SP} 2016, San Jose, CA, USA, May
  22-26, 2016, {IEEE} Computer Society, pp 470--485, \doi{10.1109/SP.2016.35},
  \urlprefix\url{http://dx.doi.org/10.1109/SP.2016.35}

\bibitem[{Escobar et~al(2009)}]{EscobarMM-fosad}
Escobar S, Meadows C, Meseguer J (2009) {M}aude-{NPA}: Cryptographic protocol
  analysis modulo equational properties. In: Foundations of Security Analysis
  and Design V, FOSAD 2007/2008/2009 Tutorial Lectures, Springer, LNCS vol.
  5705, pp 1--50

\bibitem[{Escobar et~al(2014)}]{Escobar2014FW}
Escobar S, Meadows C, Meseguer J, Santiago S (2014) A rewriting-based forwards
  semantics for {Maude-NPA}. In: Proceedings of the 2014 Symposium and Bootcamp
  on the Science of Security, HotSoS 2014, ACM

\bibitem[{Escobar et~al(2015)}]{EscobarMMS15}
Escobar S, Meadows C, Meseguer J, Santiago S (2015) Symbolic protocol analysis
  with disequality constraints modulo equational theories. In: Programming
  Languages with Applications to Biology and Security, pp 238--261

\bibitem[{Escobar et~al(2017)}]{MNPAmanual3.0}
Escobar S, Meadows C, Meseguer J (2017) Maude-NPA manual version 3.0.
  \url{http://maude.cs.illinois.edu/w/index.php?title=Maude_Tools:_Maude-NPA}

\bibitem[{Fabrega et~al(1999)}]{strands}
Fabrega FJT, Herzog J, Guttman J (1999) {Strand Spaces: What Makes a Security
  Protocol Correct?} Journal of Computer Security 7:191--230

\bibitem[{Fr{\"{o}}schle(2009)}]{DBLP:journals/entcs/Froschle09}
Fr{\"{o}}schle SB (2009) Adding branching to the strand space model. Electr
  Notes Theor Comput Sci 242(1):139--159

\bibitem[{Goubault{-}Larrecq et~al(2007)}]{DBLP:conf/aplas/Goubault-LarrecqPT07}
Goubault{-}Larrecq J, Palamidessi C, Troina A (2007) A probabilistic applied
  pi-calculus. In: Programming Languages and Systems, 5th Asian Symposium,
  {APLAS} 2007, pp 175--190

\bibitem[{Meier et~al(2013)}]{DBLP:conf/cav/MeierSCB13}
Meier S, Schmidt B, Cremers C, Basin DA (2013) The {TAMARIN} prover for the
  symbolic analysis of security protocols. In: Computer Aided Verification -
  25th International Conference, {CAV} 2013, pp 696--701

\bibitem[{Meseguer(1992)}]{Meseguer92}
Meseguer J (1992) Conditional rewriting logic as a united model of concurrency.
  Theor Comput Sci 96(1):73--155

\bibitem[{Milner(1999)}]{DBLP:books/daglib/0098267}
Milner R (1999) Communicating and mobile systems - the Pi-calculus. Cambridge
  University Press

\bibitem[{Olarte and Valencia(2008)}]{DBLP:conf/ppdp/OlarteV08}
Olarte C, Valencia FD (2008) The expressivity of universal timed {CCP:}
  undecidability of monadic {FLTL} and closure operators for security. In:
  Proceedings Principles and Practice of Declarative Programming 2008, pp 8--19

\bibitem[{Rescorla(2016)}]{tls1.3.12}
Rescorla E (2016) The transport layer security (tls) protocol version 1.3.
  Tech. Rep. draft-ietf-tls-tls13-12, IETF

\bibitem[{Yang et~al(2016)}]{YEMMS16}
Yang F, Escobar S, Meadows CA, Meseguer J, Santiago S (2016) Strand spaces with
  choice via a process algebra semantics. In: Cheney J, Vidal G (eds)
  Proceedings of the 18th International Symposium on Principles and Practice of
  Declarative Programming, Edinburgh, United Kingdom, September 5-7, 2016,
  {ACM}, pp 76--89, \doi{10.1145/2967973.2968609},
  \urlprefix\url{http://doi.acm.org/10.1145/2967973.2968609}

\end{thebibliography}

\appendix

\section{Proof of Theorem \ref{thm:Bisimulation}} \label{bisim-proof-app}

 \begin{proof}
 Since $P_{init}\rightarrow_{nil}P_{init}$ and $F_{init}\rightarrow_{nil}F_{init}$, therefore, $(P_{init}, F_{init}) \in \HState$.
We then prove that: for all PA-State $\Pst_{n}$, and FW-State $\Fst_{n}$, if $(\Pst_n, \Fst_n) \allowbreak \in \HState$, and there exists a PA-State $\Pst_{n+1}$ such that $\Pst_n \rightarrow_a \Pst_{n+1}$, then there exists a FW-State $\Fst_{n+1}$ such that $\Fst_n\rightarrow_{a} \Fst_{n+1}$ and $(\Pst_{n+1}, \Fst_{n+1})\in \HState$.. If $(\Pst_n, \Fst_n) \in \HState$,
 by definition of the relation $\HState$, there exists a label sequence $\alpha$  s.t.  $P_{init}\rightarrow_{\alpha} Pst_n$ and $F_{init} \rightarrow_{\alpha} \Fst_n$. Suppose there exists state $\Pst_{n+1}$ such that $\Pst_n \rightarrow_a \Pst_{n+1}$. We prove by case analysis on label $a$ that there exists $\Fst_{n+1}$ such that $\Fst_n\rightarrow_{a} \Fst_{n+1}$.  The fact that $(\Pst_{n+1}, \Fst_{n+1})\in \HState$ then follows this by the definition of relation $\HState$.
 
 In the rest of this proof, $\overrightarrow{L}, \overrightarrow{L_1}$ and $\overrightarrow{L_2}$ denote lists of messages, $M, M'$ and $m$ denote messages, $P, Q$ and $R$ denote processes, $PS$ denotes a process configuration, $SS$ denotes a set of constrained protocol strands, $IK$ and $IK'$ denote the set of messages in the intruder's knowledge.
 	\begin{itemize}
		\item[1)] $a=(ro, i, j, +m, 0):$ if $j>1$, according to the semantics, $\Pst_n \rightarrow_a \Pst_{n+1}$ by applying rule \eqref{eq:pa-output-modIK}, 
 		the state $\Pst_n$ is of the form $\{(ro,i, j)~( +M\cdot P) ~\&~ PS  ~|~  \{IK\} \}$ s.t. there exists a ground substitution $\sigma$ binding the choice variables in $M$ and $m=M\sigma$,
 		the state $\Pst_{n+1}= \{(ro,i, j+1)~P\sigma ~\&~ PS ~|~ \{m\in \caI, IK\}  \}$ and $ \inI{m} \notin \textit{IK}$. 
 		Since $Pst_n \ \HState \ Fst_n$, by Lemmas \ref{lemm:equiv-PStr} and \ref{lemm:sameIK}, 
 		$\Fst_n$ is of the form $\{(ro, i)~[\overrightarrow L] ~\&~ SS~\&~ \{IK\}  \}$ s.t. $(ro,i, j)~( +M\cdot P)~ \HProc ~(ro, i)~[\overrightarrow L] $. 
 		Let $(ro) \ [\overrightarrow{L_1}, \overrightarrow{L_2}]$ be a constrained strand in $\StrCstrFW$ s.t. there exists a ground substituion $\theta$ s.t. $\overrightarrow{L_1}\rho_{ro, i}\theta = \overrightarrow{L}$.
 		 By the definition of relation $\HProc$ and mapping $\toCstrFW$,  the first message of $\overrightarrow{L_2}$ is $+M'$, s.t. $M'\rho_{ro, i}\theta=M$. Then since $M\sigma =m $ and $\inI{m} \notin \textit{IK}$, the rule \eqref{eq:forward-positive-modIK-Cstr} can be applied for the rewrite $Fst_n\rightarrow_{a}Fst_{n+1}$, where $Fst_{n+1} = \{(ro, i)~[\overrightarrow L, +m] ~\&~ SS~\&~ \{\inI{m}, IK\}\}$.  

 	\ \ \ \	\ If $j=1$, $\Pst_n \rightarrow_a Pst_{n+1}$ by applying rule \eqref{eq:pa-new}, there exists a process $(ro)~( +M\cdot P) $ in $\ProcPA$ and a ground substitution $\sigma$ s.t. $M\rho_{ro, i}\sigma= m$. Since $\toCstrFW(\ProcPA)= \StrCstrFW$, by the definition of $\toCstrFW$, for all strands of role $ro$ in $\StrCstrFW$, the first message is $+M$. Without loss of generality, let $\Pst_n$ be $\{ PS  ~|~  \{IK\} \}$, and $\Fst_n$ be $\{ SS  ~\&~  \{IK'\} \}$. Since the rule \eqref{eq:pa-new} can be applied, $\inI{m} \notin \textit{IK}$. By Lemma \ref{lemm:sameIK}, $IK=IK'$. Moreover, by Lemma \ref{lemm:equiv-PStr}, $\MaxStrId(SS, ro) = \MaxProcId(PS, ro)$, and since $\MaxProcId(PS, ro) +1 = i$, by applying the rule \eqref{eq:forward-positive2-modIK-Cstr} we get $\Fst_n\rightarrow_{a}\Fst_{n+1}$.
 		 
 		\item[2)] $a=(ro, i, j, M\sigma, 0)$: similar to case 1. 
 		
 		\item[3)] $a=(ro, i, j, -m, 0)$: if $j>1$, according to the semantics, $\Pst_n \rightarrow_a \Pst_{n+1}$ by applying rule \eqref{eq:pa-input},
 		$\Pst_n$ is of the form $\{(ro, i, j)~( -M\cdot P) ~\&~ PS ~|~ \{\inI{m}, IK\} \}$ s.t. $m=_{E_{\caP}} M\sigma$ for some ground substitution $\sigma$ and 
 		$\Pst_{n+1} =\{ (ro, i, j+1)~P\sigma ~\&~ PS ~|~ \{\inI{m}, IK\}\}$. 
 	     Since $\Pst_n \ \HState \ \Fst_n$, by Lemmas \ref{lemm:equiv-PStr} and \ref{lemm:sameIK}, 
 	     $\Fst_n = \{   (ro, i)~[\overrightarrow L]  \  \&  \   SS~\&~ \{\inI{m},  IK\} \}$ s.t. $(ro,i, j)~( -M\cdot P) ~\HProc~ (ro)~[\overrightarrow L] $. 
 		Let $(ro)\ [\overrightarrow{L_1}, \overrightarrow{L_2}] \in \StrCstrFW$ s.t. 
 		there exists a ground substitution $\theta$ s.t. $\overrightarrow{L_1}\rho_{ro, i}\theta \allowbreak = \overrightarrow{L}$, 
 		then by definition of $\HProc$ and $\toCstrFW$, the first message of $\overrightarrow{L_2}$ is $-M'$ s.t. $M'\rho_{ro, i}\theta = M$. 
 		Since $m=_{E_{\caP}} M\sigma$, rule \eqref{eq:forward-negative-Cstr} can be applied to get the transition $Fst_n \rightarrow_{a} Fst_{n+1}$, 
 		where $Fst_{n+1} =  \{ (ro, i)~[\overrightarrow L,  -m] ~\&~ SS~\& ~ \{\inI{m}, IK\} \}$. 
 		
	\ \ \ \ 	\	If $j=1$, $Pst_n \rightarrow_a Pst_{n+1}$ by applying rule \eqref{eq:pa-new}, there exists a process $(ro)~( -M\cdot P) $ in $\ProcPA$ and a ground substitution $\sigma$ s.t. $M\rho_{ro, i}\sigma= m$. Without loss of generality, let $Pst_n$ be $\{ PS  ~|~  \{IK\} \}$. Then $\inI{m}  \in IK$. 
 		Since $\toCstrFW(\ProcPA)= \StrCstrFW$, by the definition of $\toCstrFW$, for all strands of role $ro$ in $\StrCstrFW$, the first message is $-M$. 
 		 By Lemma \ref{lemm:sameIK}, $\inI{m}$ is in the intruder knowledge of $Fst_n$. 
 		 Moreover, by Lemma \ref{lemm:equiv-PStr}, $\MaxStrId(SS, ro) = \MaxProcId(PS, ro)$, and since $\MaxProcId(PS, ro) +1 = i$, by applying the rule \eqref{eq:forward-negative2-Cstr} we get $Fst_n\rightarrow_{a}Fst_{n+1}$.
 		 		
 		\item[4)] $a=(ro, i, j, T, 1)$: according to the transition rules, $Pst_n \rightarrow_a Pst_{n+1}$ by applying rule \eqref{eq:pa-detBranch1}. Therefore 
 		$Pst_n$ is of the form $\{ (ro,i, j)  \ ((\textit{if} \ \ c \ \textit{then} ~P~ \textit{else}  ~Q) ~\cdot R) ~\&~ PS \mid  \{ IK\} \}$,     
 		 $Pst_{n+1} = \{(ro,i, j+1)~(P\cdot R) ~\&~ PS ~| \allowbreak ~ \{IK\} \}$  and $ c =_{E_{\caP}} true$. 
 		Since $Fst_n \ \HState \ Pst_n$, by Lemma \ref{lemm:equiv-PStr}, 
 		 $Fst_n = \{ (ro)~[\overrightarrow{L}]~\&~ SS ~ \& ~ \{IK'\}  \}$ 
 		 s.t. $ (ro,i, j)~((\textit{if} ~c~ \allowbreak \textit{then} ~P~ \textit{else} ~Q) ~\cdot R) ~\HProc~ (ro, i)~[\overrightarrow{L}] $. 
 		 By the definition of the relation $\HProc$ and the mapping $\toCstrFW$,
 		  there exists $(ro) \ [\overrightarrow{L_1}, \{C, 1\}, \allowbreak \overrightarrow{L_2}] \in \StrCstrFW$ and a ground substitution $\theta$ s.t. $\overrightarrow{L} = \overrightarrow{L_1}\rho_{ro, i}\theta$, and $C\rho_{ro,i}\theta=c$. Since $c =_{E_{\caP}} true$, the rule \eqref{eq:forward-CstrMsg-det} can be applied for the rewrite $Fst_n \rightarrow_a Fst_{n+1}$,  where $Fst_{n+1}= \{
\{ (ro)~[\overrightarrow{L}, \allowbreak\{t, 1\}]~ \&~ SS ~ \& ~ \{IK'\} \}$
 		\item[5)] $a=(ro, i, j, T, 2)$: similar to case 4.
 		\item[6)] $a=(ro, i, j, ?, 1)$: if $j>1$, $Pst_n \rightarrow_a Pst_{n+1}$ by applying rule \eqref{eq:pa-nonDetBranch1}. Therefore 
 		$Pst_n $ is of the form $\{(ro,i, j)~((P~?~Q)\cdot R) ~\&~ PS ~|~ \allowbreak\{IK\} \}$  and
 		 $Pst_{n+1} = \{(ro,i, j+1)~(P \cdot R) ~\&~ PS ~|~ \{IK\} \}$. 
 		 Since $Fst_n \ \HState \ Pst_n$, by Lemma \ref{lemm:equiv-PStr}, 
 		 $Fst_n= \{ (ro, i)~[\overrightarrow{L}]~\&~ SS ~ \& ~ \{IK'\}  \}$ 
 		 s.t. $  (ro, i, j)~((P~?~Q)\cdot R) ~\HProc (ro, i)~[\overrightarrow{L}] $. 
 		 By the definition of $\HProc$ and $\toCstrFW$, there is a strand 
 		 $(ro, i) \ [\overrightarrow{L_1}, \{?, 1\}, \overrightarrow{L_2}] \in \StrCstrFW$ s.t. 
 		 $\overrightarrow{L} = \overrightarrow{L_1}\theta$. 
 		 Therefore, rule \eqref{eq:forward-CstrMsg-nonDet} can be applied for the rewrite $Fst_n \rightarrow_a Fst_{n+1}$, 
 		 and $Fst_{n+1}  =  \{ (ro, i)~[\overrightarrow{L}, \allowbreak \{?, 1\}]~\&~ SS \allowbreak ~ \& ~ \{IK'\}  \}$. 
 		 
 	\ \ \ \ \		 If $j=1$, $Pst_n \rightarrow_a Pst_{n+1}$ by applying rule \eqref{eq:pa-new}.
 Therefore, there exists a process $(ro)~( (P~?~Q)\cdot R) $ in $\ProcPA$.
 		Since $\toCstrFW(\ProcPA)= \StrCstrFW$, by the definition of $\toCstrFW$, there is a strand of role $ro$ whose first message is $(?, 1)$ in $\StrCstrFW$.
 		 Moreover, by Lemma \ref{lemm:equiv-PStr}, $\MaxStrId(SS, ro) = \MaxProcId(PS, ro)$, and since $\MaxProcId(PS, ro) +1 = i$, by applying the rule \eqref{eq:forward-CstrMsg-first} we get $Fst_n\rightarrow_{a}Fst_{n+1}$.
 		 		
 		\item[7)] $a=(ro, i, j, ?, 2)$ similar to case 6.
 	\end{itemize}
 Similarly, we can prove that for all PA-State $\Pst_{n}$, and FW-State $\Fst_{n}$, if $(\Pst_n, \Fst_n) \allowbreak \in \HState$, and there exists a FW-State $\Fst_{n+1}$ such that $\Fst_n \rightarrow_a \Fst_{n+1}$, then there exists a PA-State $\Pst_{n+1}$ such that $\Pst_n\rightarrow_{a} \Pst_{n+1}$ and $(\Pst_{n+1}, \Fst_{n+1})\in \HState$ 
 \end{proof}

\end{document}